\documentclass[runningheads]{llncs}

\usepackage{amssymb}
\usepackage{amsmath}
\usepackage{stmaryrd}
\usepackage{enumitem}
\usepackage{todonotes}
\usepackage[all]{xy}
\usepackage{xcolor}
\usepackage{proof}
\usepackage{wrapfig}

\newif\ifmarkers
\markersfalse 
\newif\ifextended
\extendedtrue 
\newif\ifreview
\ifmarkers
 \reviewfalse
\else
 \reviewtrue
\fi

\ifmarkers
 \usepackage{showlabels}
\fi

\ifextended
\usepackage{apxproof}
\else
\usepackage[appendix=strip]{apxproof}
\fi

\def\bra#1{\mathinner{\langle{#1}|}}

\usepackage[rflt]{floatflt}

\newcommand{\Bf}[1]{{\bf #1}}

\newcommand{\Rm}[1]{{\rm #1}}
\newcommand{\ul}{\underline}
\newcommand{\ol}{\overline}

\newcommand{\ex}[1]{\exists{#1}~.~}

\newcommand{\lam}[1]{\lambda{#1}~.~}

\newcommand{\lsem}{[\![}
\newcommand{\rsem}{]\!]}
\newcommand{\sem}[1]{\lsem #1\rsem}

\newcommand{\op}{\mathit{op}}
\newcommand{\dom}{\Rm{dom}}
\newcommand{\cod}{\Rm{cod}}

\newcommand{\Obj}[1]{\Bf{Obj}(#1)}

\newcommand{\Set}{\Bf{Set}}

\newcommand{\Pred}{\Bf{Pred}}

\newcommand{\CAT}{\Bf{CAT}}

\newcommand{\BRel}{\Bf{BRel}}

\newcommand{\arrow}{\rightarrow}

\newcommand{\Id}{\Rm{Id}}
\newcommand{\id}{\Rm{id}}

\newcommand{\CC}{\Bbb C}

\newcommand{\EE}{\Bbb E}

\newcommand{\NN}{\Bbb N}
\newcommand{\PP}{\Bbb P}
\newcommand{\RR}{\Bbb R}

\newcommand{\VV}{\Bbb V}

\newcommand{\darrow}{\mathbin{\dot\arrow}}

\newcommand{\dtimes}{\mathbin{\dot\times}}

\newcommand{\rue}{\ar@{=}[u]}
\newcommand{\rueh}[1]{\ar@{=}[u]^-{#1}}
\newcommand{\ruem}[1]{\ar@{=}[u]_-{#1}}
\newcommand{\ruen}[1]{\ar@{=}[u]|-{#1}}

\newcommand{\ruue}{\ar@{=}[uu]}
\newcommand{\ruueh}[1]{\ar@{=}[uu]^-{#1}}
\newcommand{\ruuem}[1]{\ar@{=}[uu]_-{#1}}
\newcommand{\ruuen}[1]{\ar@{=}[uu]|-{#1}}

\newcommand{\ruuue}{\ar@{=}[uuu]}
\newcommand{\ruuueh}[1]{\ar@{=}[uuu]^-{#1}}
\newcommand{\ruuuem}[1]{\ar@{=}[uuu]_-{#1}}
\newcommand{\ruuuen}[1]{\ar@{=}[uuu]|-{#1}}
\newcommand{\rd}{\ar[d]}
\newcommand{\rdh}[1]{\ar[d]^-{#1}}
\newcommand{\rdm}[1]{\ar[d]_-{#1}}

\newcommand{\rde}{\ar@{=}[d]}
\newcommand{\rdeh}[1]{\ar@{=}[d]^-{#1}}
\newcommand{\rdem}[1]{\ar@{=}[d]_-{#1}}
\newcommand{\rden}[1]{\ar@{=}[d]|-{#1}}

\newcommand{\rdde}{\ar@{=}[dd]}
\newcommand{\rddeh}[1]{\ar@{=}[dd]^-{#1}}
\newcommand{\rddem}[1]{\ar@{=}[dd]_-{#1}}
\newcommand{\rdden}[1]{\ar@{=}[dd]|-{#1}}

\newcommand{\rddde}{\ar@{=}[ddd]}
\newcommand{\rdddeh}[1]{\ar@{=}[ddd]^-{#1}}
\newcommand{\rdddem}[1]{\ar@{=}[ddd]_-{#1}}
\newcommand{\rddden}[1]{\ar@{=}[ddd]|-{#1}}
\newcommand{\rr}{\ar[r]}
\newcommand{\rrh}[1]{\ar[r]^-{#1}}

\newcommand{\rre}{\ar@{=}[r]}
\newcommand{\rreh}[1]{\ar@{=}[r]^-{#1}}
\newcommand{\rrem}[1]{\ar@{=}[r]_-{#1}}
\newcommand{\rren}[1]{\ar@{=}[r]|-{#1}}

\newcommand{\rrue}{\ar@{=}[ru]}
\newcommand{\rrueh}[1]{\ar@{=}[ru]^-{#1}}
\newcommand{\rruem}[1]{\ar@{=}[ru]_-{#1}}
\newcommand{\rruen}[1]{\ar@{=}[ru]|-{#1}}

\newcommand{\rruue}{\ar@{=}[ruu]}
\newcommand{\rruueh}[1]{\ar@{=}[ruu]^-{#1}}
\newcommand{\rruuem}[1]{\ar@{=}[ruu]_-{#1}}
\newcommand{\rruuen}[1]{\ar@{=}[ruu]|-{#1}}

\newcommand{\rruuue}{\ar@{=}[ruuu]}
\newcommand{\rruuueh}[1]{\ar@{=}[ruuu]^-{#1}}
\newcommand{\rruuuem}[1]{\ar@{=}[ruuu]_-{#1}}
\newcommand{\rruuuen}[1]{\ar@{=}[ruuu]|-{#1}}
\newcommand{\rrd}{\ar[rd]}

\newcommand{\rrde}{\ar@{=}[rd]}
\newcommand{\rrdeh}[1]{\ar@{=}[rd]^-{#1}}
\newcommand{\rrdem}[1]{\ar@{=}[rd]_-{#1}}
\newcommand{\rrden}[1]{\ar@{=}[rd]|-{#1}}

\newcommand{\rrdde}{\ar@{=}[rdd]}
\newcommand{\rrddeh}[1]{\ar@{=}[rdd]^-{#1}}
\newcommand{\rrddem}[1]{\ar@{=}[rdd]_-{#1}}
\newcommand{\rrdden}[1]{\ar@{=}[rdd]|-{#1}}

\newcommand{\rrddde}{\ar@{=}[rddd]}
\newcommand{\rrdddeh}[1]{\ar@{=}[rddd]^-{#1}}
\newcommand{\rrdddem}[1]{\ar@{=}[rddd]_-{#1}}
\newcommand{\rrddden}[1]{\ar@{=}[rddd]|-{#1}}
\newcommand{\rrr}{\ar[rr]}
\newcommand{\rrrh}[1]{\ar[rr]^-{#1}}

\newcommand{\rrre}{\ar@{=}[rr]}
\newcommand{\rrreh}[1]{\ar@{=}[rr]^-{#1}}
\newcommand{\rrrem}[1]{\ar@{=}[rr]_-{#1}}
\newcommand{\rrren}[1]{\ar@{=}[rr]|-{#1}}

\newcommand{\rrrue}{\ar@{=}[rru]}
\newcommand{\rrrueh}[1]{\ar@{=}[rru]^-{#1}}
\newcommand{\rrruem}[1]{\ar@{=}[rru]_-{#1}}
\newcommand{\rrruen}[1]{\ar@{=}[rru]|-{#1}}

\newcommand{\rrruue}{\ar@{=}[rruu]}
\newcommand{\rrruueh}[1]{\ar@{=}[rruu]^-{#1}}
\newcommand{\rrruuem}[1]{\ar@{=}[rruu]_-{#1}}
\newcommand{\rrruuen}[1]{\ar@{=}[rruu]|-{#1}}

\newcommand{\rrruuue}{\ar@{=}[rruuu]}
\newcommand{\rrruuueh}[1]{\ar@{=}[rruuu]^-{#1}}
\newcommand{\rrruuuem}[1]{\ar@{=}[rruuu]_-{#1}}
\newcommand{\rrruuuen}[1]{\ar@{=}[rruuu]|-{#1}}

\newcommand{\rrrde}{\ar@{=}[rrd]}
\newcommand{\rrrdeh}[1]{\ar@{=}[rrd]^-{#1}}
\newcommand{\rrrdem}[1]{\ar@{=}[rrd]_-{#1}}
\newcommand{\rrrden}[1]{\ar@{=}[rrd]|-{#1}}

\newcommand{\rrrdde}{\ar@{=}[rrdd]}
\newcommand{\rrrddeh}[1]{\ar@{=}[rrdd]^-{#1}}
\newcommand{\rrrddem}[1]{\ar@{=}[rrdd]_-{#1}}
\newcommand{\rrrdden}[1]{\ar@{=}[rrdd]|-{#1}}

\newcommand{\rrrddde}{\ar@{=}[rrddd]}
\newcommand{\rrrdddeh}[1]{\ar@{=}[rrddd]^-{#1}}
\newcommand{\rrrdddem}[1]{\ar@{=}[rrddd]_-{#1}}
\newcommand{\rrrddden}[1]{\ar@{=}[rrddd]|-{#1}}

\newcommand{\rrrre}{\ar@{=}[rrr]}
\newcommand{\rrrreh}[1]{\ar@{=}[rrr]^-{#1}}
\newcommand{\rrrrem}[1]{\ar@{=}[rrr]_-{#1}}
\newcommand{\rrrren}[1]{\ar@{=}[rrr]|-{#1}}

\newcommand{\rrrrue}{\ar@{=}[rrru]}
\newcommand{\rrrrueh}[1]{\ar@{=}[rrru]^-{#1}}
\newcommand{\rrrruem}[1]{\ar@{=}[rrru]_-{#1}}
\newcommand{\rrrruen}[1]{\ar@{=}[rrru]|-{#1}}

\newcommand{\rrrruue}{\ar@{=}[rrruu]}
\newcommand{\rrrruueh}[1]{\ar@{=}[rrruu]^-{#1}}
\newcommand{\rrrruuem}[1]{\ar@{=}[rrruu]_-{#1}}
\newcommand{\rrrruuen}[1]{\ar@{=}[rrruu]|-{#1}}

\newcommand{\rrrruuue}{\ar@{=}[rrruuu]}
\newcommand{\rrrruuueh}[1]{\ar@{=}[rrruuu]^-{#1}}
\newcommand{\rrrruuuem}[1]{\ar@{=}[rrruuu]_-{#1}}
\newcommand{\rrrruuuen}[1]{\ar@{=}[rrruuu]|-{#1}}

\newcommand{\rrrrde}{\ar@{=}[rrrd]}
\newcommand{\rrrrdeh}[1]{\ar@{=}[rrrd]^-{#1}}
\newcommand{\rrrrdem}[1]{\ar@{=}[rrrd]_-{#1}}
\newcommand{\rrrrden}[1]{\ar@{=}[rrrd]|-{#1}}

\newcommand{\rrrrdde}{\ar@{=}[rrrdd]}
\newcommand{\rrrrddeh}[1]{\ar@{=}[rrrdd]^-{#1}}
\newcommand{\rrrrddem}[1]{\ar@{=}[rrrdd]_-{#1}}
\newcommand{\rrrrdden}[1]{\ar@{=}[rrrdd]|-{#1}}

\newcommand{\rrrrddde}{\ar@{=}[rrrddd]}
\newcommand{\rrrrdddeh}[1]{\ar@{=}[rrrddd]^-{#1}}
\newcommand{\rrrrdddem}[1]{\ar@{=}[rrrddd]_-{#1}}
\newcommand{\rrrrddden}[1]{\ar@{=}[rrrddd]|-{#1}}

\newcommand{\rle}{\ar@{=}[l]}
\newcommand{\rleh}[1]{\ar@{=}[l]^-{#1}}
\newcommand{\rlem}[1]{\ar@{=}[l]_-{#1}}
\newcommand{\rlen}[1]{\ar@{=}[l]|-{#1}}

\newcommand{\rlue}{\ar@{=}[lu]}
\newcommand{\rlueh}[1]{\ar@{=}[lu]^-{#1}}
\newcommand{\rluem}[1]{\ar@{=}[lu]_-{#1}}
\newcommand{\rluen}[1]{\ar@{=}[lu]|-{#1}}

\newcommand{\rluue}{\ar@{=}[luu]}
\newcommand{\rluueh}[1]{\ar@{=}[luu]^-{#1}}
\newcommand{\rluuem}[1]{\ar@{=}[luu]_-{#1}}
\newcommand{\rluuen}[1]{\ar@{=}[luu]|-{#1}}

\newcommand{\rluuue}{\ar@{=}[luuu]}
\newcommand{\rluuueh}[1]{\ar@{=}[luuu]^-{#1}}
\newcommand{\rluuuem}[1]{\ar@{=}[luuu]_-{#1}}
\newcommand{\rluuuen}[1]{\ar@{=}[luuu]|-{#1}}

\newcommand{\rlde}{\ar@{=}[ld]}
\newcommand{\rldeh}[1]{\ar@{=}[ld]^-{#1}}
\newcommand{\rldem}[1]{\ar@{=}[ld]_-{#1}}
\newcommand{\rlden}[1]{\ar@{=}[ld]|-{#1}}

\newcommand{\rldde}{\ar@{=}[ldd]}
\newcommand{\rlddeh}[1]{\ar@{=}[ldd]^-{#1}}
\newcommand{\rlddem}[1]{\ar@{=}[ldd]_-{#1}}
\newcommand{\rldden}[1]{\ar@{=}[ldd]|-{#1}}

\newcommand{\rlddde}{\ar@{=}[lddd]}
\newcommand{\rldddeh}[1]{\ar@{=}[lddd]^-{#1}}
\newcommand{\rldddem}[1]{\ar@{=}[lddd]_-{#1}}
\newcommand{\rlddden}[1]{\ar@{=}[lddd]|-{#1}}

\newcommand{\rlle}{\ar@{=}[ll]}
\newcommand{\rlleh}[1]{\ar@{=}[ll]^-{#1}}
\newcommand{\rllem}[1]{\ar@{=}[ll]_-{#1}}
\newcommand{\rllen}[1]{\ar@{=}[ll]|-{#1}}

\newcommand{\rllue}{\ar@{=}[llu]}
\newcommand{\rllueh}[1]{\ar@{=}[llu]^-{#1}}
\newcommand{\rlluem}[1]{\ar@{=}[llu]_-{#1}}
\newcommand{\rlluen}[1]{\ar@{=}[llu]|-{#1}}

\newcommand{\rlluue}{\ar@{=}[lluu]}
\newcommand{\rlluueh}[1]{\ar@{=}[lluu]^-{#1}}
\newcommand{\rlluuem}[1]{\ar@{=}[lluu]_-{#1}}
\newcommand{\rlluuen}[1]{\ar@{=}[lluu]|-{#1}}

\newcommand{\rlluuue}{\ar@{=}[lluuu]}
\newcommand{\rlluuueh}[1]{\ar@{=}[lluuu]^-{#1}}
\newcommand{\rlluuuem}[1]{\ar@{=}[lluuu]_-{#1}}
\newcommand{\rlluuuen}[1]{\ar@{=}[lluuu]|-{#1}}

\newcommand{\rllde}{\ar@{=}[lld]}
\newcommand{\rlldeh}[1]{\ar@{=}[lld]^-{#1}}
\newcommand{\rlldem}[1]{\ar@{=}[lld]_-{#1}}
\newcommand{\rllden}[1]{\ar@{=}[lld]|-{#1}}

\newcommand{\rlldde}{\ar@{=}[lldd]}
\newcommand{\rllddeh}[1]{\ar@{=}[lldd]^-{#1}}
\newcommand{\rllddem}[1]{\ar@{=}[lldd]_-{#1}}
\newcommand{\rlldden}[1]{\ar@{=}[lldd]|-{#1}}

\newcommand{\rllddde}{\ar@{=}[llddd]}
\newcommand{\rlldddeh}[1]{\ar@{=}[llddd]^-{#1}}
\newcommand{\rlldddem}[1]{\ar@{=}[llddd]_-{#1}}
\newcommand{\rllddden}[1]{\ar@{=}[llddd]|-{#1}}

\newcommand{\rllle}{\ar@{=}[lll]}
\newcommand{\rllleh}[1]{\ar@{=}[lll]^-{#1}}
\newcommand{\rlllem}[1]{\ar@{=}[lll]_-{#1}}
\newcommand{\rlllen}[1]{\ar@{=}[lll]|-{#1}}

\newcommand{\rlllue}{\ar@{=}[lllu]}
\newcommand{\rlllueh}[1]{\ar@{=}[lllu]^-{#1}}
\newcommand{\rllluem}[1]{\ar@{=}[lllu]_-{#1}}
\newcommand{\rllluen}[1]{\ar@{=}[lllu]|-{#1}}

\newcommand{\rllluue}{\ar@{=}[llluu]}
\newcommand{\rllluueh}[1]{\ar@{=}[llluu]^-{#1}}
\newcommand{\rllluuem}[1]{\ar@{=}[llluu]_-{#1}}
\newcommand{\rllluuen}[1]{\ar@{=}[llluu]|-{#1}}

\newcommand{\rllluuue}{\ar@{=}[llluuu]}
\newcommand{\rllluuueh}[1]{\ar@{=}[llluuu]^-{#1}}
\newcommand{\rllluuuem}[1]{\ar@{=}[llluuu]_-{#1}}
\newcommand{\rllluuuen}[1]{\ar@{=}[llluuu]|-{#1}}

\newcommand{\rlllde}{\ar@{=}[llld]}
\newcommand{\rllldeh}[1]{\ar@{=}[llld]^-{#1}}
\newcommand{\rllldem}[1]{\ar@{=}[llld]_-{#1}}
\newcommand{\rlllden}[1]{\ar@{=}[llld]|-{#1}}

\newcommand{\rllldde}{\ar@{=}[llldd]}
\newcommand{\rlllddeh}[1]{\ar@{=}[llldd]^-{#1}}
\newcommand{\rlllddem}[1]{\ar@{=}[llldd]_-{#1}}
\newcommand{\rllldden}[1]{\ar@{=}[llldd]|-{#1}}

\newcommand{\rlllddde}{\ar@{=}[lllddd]}
\newcommand{\rllldddeh}[1]{\ar@{=}[lllddd]^-{#1}}
\newcommand{\rllldddem}[1]{\ar@{=}[lllddd]_-{#1}}
\newcommand{\rlllddden}[1]{\ar@{=}[lllddd]|-{#1}}

\newenvironment{choice}{\left\{\begin{array}{ll}}{\end{array}\right.}

\newcommand{\adjunction}[3]{
  \ar@<.4pc>[#1]^-{#2} \ar@{}[#1]|-*=0[@]{\bot} \ar@<-.4pc>@{<-}[#1]_-{#3}
}

\ifreview
\newcommand{\dnote}[1]{}
\newcommand{\snote}[1]{}
\newcommand{\mnote}[1]{}
\newcommand{\tnote}[1]{}
\newcommand{\draft}[2][-]{#2}
\else
\newcommand{\dnote}[1]{\textcolor{purple}{Dominic: #1}}
\newcommand{\snote}[1]{\textcolor{blue}{Shin-ya: #1}}
\newcommand{\mnote}[1]{\textcolor{red}{Marco: #1}}
\newcommand{\tnote}[1]{\textcolor{teal}{Tetsuya: #1}}
\newcommand{\draft}[2][-]{\marginpar{\textcolor{blue}{(#1)}}\textcolor{blue}{#2}}
\fi

\definecolor{RED}{rgb}{1,0,0}
\newcommand{\conf}[1]{} 

\let\olddefinition\definition
\renewcommand{\definition}{\olddefinition\normalfont}
\spnewtheorem{defn}[definition]{Definition}{\bfseries}{}

\newcommand{\binop}[2]{#1 \; \mathtt{op} \; #2}

\newcommand{\langNameAlone}{\red{loop}}
\newcommand{\Lang}{\red{the \langNameAlone{} language}}

\newcommand{\LangTitleTitle}{\red{The Loop Language}}

\newcommand{\Var}{\Bf{Var}}
\newcommand{\CExp}{\Bf{CExp}}
\newcommand{\PExp}{\Bf{PExp}}

\newcommand{\Axiom}{\Bf{Axiom}}

\newcommand{\sbool}{\mathtt{bool}}
\newcommand{\snat}{\mathtt{nat}}
\newcommand{\scell}{\mathtt{cell}}
\newcommand{\exprc}{e} 
\newcommand{\otrue}{\mathtt{tt}}
\newcommand{\ofalse}{\mathtt{ff}}
\newcommand{\onat}[1]{\lceil #1\rceil}

\newcommand{\Ctx}[1]{\Bf{Ctx}_{#1}}
\newcommand{\ExpSet}[1]{\Bf{Exp}_{#1}}
\newcommand{\Exp}[3]{\ExpSet{#1}(#2,#3)}
\newcommand{\Eq}{\mathrm{Eq}}

\newcommand{\mctx}{\Gamma_\mem}
\newcommand{\lsig}{\Sigma_l}

\newcommand{\synSkip}{\texttt{skip}}
\newcommand{\synSeq}[2]{#1 \mathbin{;} #2}
\newcommand{\synExp}[2]{#1 \mathbin{:=}#2}
\newcommand{\synComm}[1]{\mathop{\texttt{do}}#1}
\newcommand{\synProc}[2]{\mathop{\texttt{do}}{#1\leftarrow #2}}
\newcommand{\synIf}[3]{\mathop{\texttt{if}}#1\mathbin{\texttt{then}} #2 \mathbin{\texttt{else}} #3}
\newcommand{\synLoop}[2]{\mathop{\texttt{loop}} #1 \mathbin{\texttt{do}} #2}
\newcommand{\osample}[2]{\mathtt{sample}_{#1,#2}}

\newcommand{\subst}[2]{[#1/#2]}

\newcommand{\tick}{\mathtt{tick}}
\newcommand{\cfTrue}{\mathtt{cfTT}}
\newcommand{\cfFalse}{\mathtt{cfFF}}

\newcommand{\Psecure}[1]{\mathtt{secr}_{#1}}
\newcommand{\SecLV}{\mathtt{SecLV}}

\newcommand{\VarLV}{\mathsf{VarLV}}

\newcommand{\interp}[1]{\sem{#1}}

\newcommand{\Fml}[2]{\Bf{Fml}_{#1}(#2)}

\newcommand{\bool}{\mathsf{Bool}}
\newcommand{\nat}{\mathsf{Nat}}
\newcommand{\push}[1]{\mathsf{ext}{#1}}
\newcommand{\mtrue}{\mathsf{tt}}
\newcommand{\mfalse}{\mathsf{ff}}
\newcommand{\mnat}[1]{\lfloor#1\rfloor}

\newcommand{\upd}[1]{\mathsf{upd}_{#1}} 
\newcommand{\sub}[2]{\mathsf{sub}(#1,#2)} 

\newcommand{\ftrue}{\mathsf{tm}}
\newcommand{\ffalse}{\mathsf{fm}}
\newcommand{\fnat}[1]{[#1]}

\newcommand{\mem}{\mathsf{M}}

\newcommand{\tripleV}[4]{\vdash_{#1} \{#2\} \mathbin{#3} \{#4\}}

\newcommand{\transpose}[1]{#1^{\mathsf{T}}}

\newcommand{\garrow}[1]{\mathbin{\arrow_{#1}}}
\newcommand{\inverse}[1]{{#1}^{\hspace{-0.15em} - \hspace{-0.1em}1}\hspace{-0.1em}}

\newcommand{\dto}{\mathbin{\dot\arrow}}

\newcommand{\ERel}{\Bf{ERel}}
\newcommand{\Meas}{\Bf{Meas}}

\newcommand{\teq}{\mathbin{\triangleq}}

\newcommand{\lev}{\mathbin{\nearrow}}

\newcommand{\pbcorner}[1][dr]{\save*!/#1+1.2pc/#1:(1,-1)@^{|-}\restore}

\newcommand{\predfib}[1]{{p^{#1}}}
\newcommand{\brelfib}[1]{{r^{#1}}}
\newcommand{\erelfib}[1]{{e^{#1}}}

\newcommand{\natrep}{\mathbb{N}} 

\ifmarkers
\newcommand{\red}[1]{\textcolor{red}{#1}}
\else
\newcommand{\red}[1]{#1}
\fi

\newcommand{\ghlstr}{\red{GHL structure}}
\newcommand{\ghlstrs}{\red{GHL structures}}
\newcommand{\Ghlstr}{\red{GHL Structure}}
\newcommand{\pmonoid}{\red{pomonoid}}
\newcommand{\blank}{-}

\usepackage{graphicx}
\makeatletter
\RequirePackage[bookmarks,unicode,colorlinks=true]{hyperref}%
   \def\@citecolor{blue}%
   \def\@urlcolor{blue}%
   \def\@linkcolor{blue}%

\def\orcidID#1{\smash{\href{http://orcid.org/#1}{\protect\raisebox{-1.25pt}{\protect\includegraphics{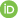}}}}}
\makeatother
\usepackage{marvosym} 

\begin{document}
\title{Graded Hoare Logic and its Categorical Semantics}
\titlerunning{Graded Hoare Logic and its Categorical Semantics}
\authorrunning{Marco Gaboardi, Shin-ya Katsumata, Dominic Orchard, Tetsuya Sato}

\author{
Marco Gaboardi\inst{1}\textsuperscript{(\Letter)}
\and
Shin{-}ya Katsumata\inst{2}\orcidID{0000-0001-7529-5489}
\and
Dominic Orchard\inst{3}\orcidID{0000-0002-7058-7842}
\and
Tetsuya Sato\inst{4}\orcidID{0000-0001-9895-9209} 
}
\institute{
Boston University, Boston, USA
\email{gaboardi@bu.edu}
\and
National Institute of Informatics, Tokyo, Japan
\email{s-katsumata@nii.ac.jp}
\and
University of Kent, Canterbury, United Kingdom
\email{d.a.orchard@kent.ac.uk}
\and
Tokyo Institute of Technology, Tokyo, Japan
\email{tsato@c.titech.ac.jp}
}

\maketitle

\begin{abstract}
  Deductive verification techniques based on program logics (i.e.,
  the family of Floyd-Hoare logics) are a powerful approach for
  program reasoning. Recently, there has been a trend of increasing
  the expressive power of such logics by augmenting their rules with
  additional information to reason about program
  side-effects. For example, general program logics have been
  augmented with cost analyses, logics for probabilistic computations
  have been augmented with estimate measures, and logics for
  differential privacy with indistinguishability bounds. In this work, we unify
  these various approaches via the paradigm of
  \emph{grading}, adapted from the world of functional calculi and
  semantics. We propose \emph{Graded Hoare Logic} (GHL), a
  parameterisable framework for augmenting program logics with a
  preordered monoidal analysis. We develop a semantic framework for
  modelling GHL such that grading, logical assertions (pre- and
  post-conditions) and the underlying effectful semantics of an
  imperative language can be integrated together. Central to our
  framework is the notion of a \emph{graded category} which we extend
  here, introducing \emph{graded Freyd categories} which provide a
  semantics that can interpret many examples of augmented program
  logics from the literature. We leverage coherent fibrations to model
  the base assertion language, and thus the overall setting is also
  fibrational.
\end{abstract}


\ifextended
\begin{toappendix}
\section{Further examples}
\label{sec:further-examples}

\begin{example}[Dataflow-aware Bounded Reading] 
  \label{exm:dataflow}
  Via our grading, we can implement classical-style dataflow analyses.
  For simplicity, we consider a counting analysis
  akin to the usage analysis in Bounded Linear
  Logic~\cite{girard1992bounded}, where instead of counting
  the number of times a variable is used, we count the number of
  times that variables are read from, in a dataflow-aware fashion.

  Let $N$ be the number of variables in $\mctx$.  An expression
  analysis is a function
  $\mathsf{count}_s : \Exp\Sigma\mctx s \rightarrow \mathbb{N}^{N}$
  ($s\in S$) which computes an $N$-vector whose elements correspond to
  the number of times that each program variable is read by an
  expression:
  \begin{align}
    \label{eq:countvar}
    \countVar{v_i} = \bra{i} \quad
    \countVar{o(e_1,\cdots,e_n)} = \countVar{e_1} +\cdots+ \countVar{e_n}
  \end{align}
  where $v_i$ is the $i^{th}$ variable in $\mctx$ and $\bra{i}$
  computes a single-entry vector where
   $\bra{i}_j = 1$ if $i=j$ otherwise $\bra{i}_j = 0$.
   For example, in a program with three variables $x$, $y$, and
   $z$ (taken in that order for the analysis), then
   $\countVar{x + y + y} = \langle{1, 2, 0}\rangle$.

   We define a dataflow-aware analysis by grading GHL
   with a \pmonoid{} $\mathcal{M}$ of square matrices
   $N \times N$ with the all-zeros matrix $\bar{0}_{N \times N}$
   as the unit, preordering $\leq$ pointwise on $\mathbb{N}$,
   and multiplication:
  \begin{equation*}
    A \fatsemi B \teq (B \ast A) + A + B.
  \end{equation*}
  To integrate the analysis to GHL, we analyse assignments via
  procedures such that $\synProc{v}{\#(e)}$ denotes
  the assignment of an expression $e$ of sort $\mctx(v)$ to variable $v$ incurring
  an analysis in the GHL:
  \begin{equation*}
    \PExp_s=\{\#(e)~|~e:s\},
    \quad
    \#(e) \in
    \availproc s(\psi, \transpose{\bra{i}} \ast \countVar{e}, r = e)
  \end{equation*}
  where the analysis of an assignment to $v_i$ of expression $e$
  is given by the grade $\transpose{\bra{i}}
  \ast \countVar{e}$ which multiplies the transpose of
  the single-entry basis vector for $i$ with the analysis
  $\countVar{e}$, yielding an $N \times N$ data flow matrix.
  For example, take the program $\synProc{x}{\#(y + 2)};
  \synProc{z}{\#(x + y)}$. The first assignment has grade:
  \setlength\arraycolsep{3pt}
  \begin{align*}
    \transpose{\bra{x}} \ast \countVar{y+2} =
    \begin{bmatrix} 1 \\ 0 \\ 0 \end{bmatrix}
    \ast
    \begin{bmatrix} 0 & 1 & 0 \end{bmatrix} =
                            \begin{bmatrix}
                              0 & 1 & 0 \\
                              0 & 0 & 0 \\
                              0 & 0 & 0
                            \end{bmatrix}
\end{align*}
The resulting matrix can read as follows: the columns correspond to the
source of the dataflow and the rows the target, thus
in this case we see that $y$ (second variable) is used once to form
$x$ (first variable).
The second assignment has the analysis:
\begin{align*}
  \transpose{\bra{z}} \ast \countVar{x+y} =
  \begin{bmatrix} 0 \\ 0 \\ 1 \end{bmatrix}
  \ast
  \begin{bmatrix} 1 & 0 & 1 \end{bmatrix} =
    \begin{bmatrix}
      0 & 0 & 0 \\
      0 & 0 & 0 \\
      1 & 1 & 0
    \end{bmatrix}
\end{align*}
Taken together using $\fatsemi$ from the sequential composition, the
full program then has the analysis:
\begin{align*}
 & \begin{bmatrix}
      0 & 1 & 0 \\
      0 & 0 & 0 \\
      0 & 0 & 0
    \end{bmatrix}
\fatsemi
    \begin{bmatrix}
      0 & 0 & 0 \\
      0 & 0 & 0 \\
      1 & 1 & 0
    \end{bmatrix}
= \;
  \begin{bmatrix}
      0 & 0 & 0 \\
      0 & 0 & 0 \\
      1 & 1 & 0
    \end{bmatrix}
\!\!\ast\!\!
    \begin{bmatrix}
      0 & 1 & 0 \\
      0 & 0 & 0 \\
      0 & 0 & 0
    \end{bmatrix}
+
\begin{bmatrix}
      0 & 1 & 0 \\
      0 & 0 & 0 \\
      0 & 0 & 0
    \end{bmatrix}
+
    \begin{bmatrix}
      0 & 0 & 0 \\
      0 & 0 & 0 \\
      1 & 1 & 0
    \end{bmatrix}
\; =
\begin{bmatrix}
      0 & 1 & 0 \\
      0 & 0 & 0 \\
      1 & 2 & 0
    \end{bmatrix}
\end{align*}
Thus, $y$ flows once to $x$ and twice to $z$ and $x$ flows once to
$z$. The analysis of each assignment gives an adjacency
matrix for the dataflow graph of the statement and $\fatsemi$ computes the two-hop paths
between two dataflow graphs $(B \ast A)$ plus
the original adjacency/flows $A$ and $B$ for each statement.


\paragraph{Semantic model in our framework}

  Let $\Var$ be the set of variables in $\mctx$ and $M$ be the set of
  functions $\Var\times\Var\arrow\mathbb N$, regarded as a square
  matrix over $\mathbb N$.  We define a binary operation $\fatsemi$ on
  $M$ by
  $(A \fatsemi B)(x,y) = A(x,y) + B(x,y) + \sum_{z \in \Var} A(x,z)
  \cdot B(z,y)$.  The operation $\fatsemi$ is associative, and
  everywhere-0 matrix $\emptyset$ is the unit of this operation.  The
  componentwise numerical order between matrices makes
  $(M,\le,\emptyset,\fatsemi)$ a partially ordered monoid.

  We now invoke Theorem \ref{theorem:logical_structure:graded_lifting} by
  letting $p$ be the fibration $\predfib{\Set}:\Pred\to\Set$, $T$ be
  the $M$-valued writer monad $W_M X = X\times M$ and $\dot T$ be the
  strong $M$-graded lifting $\dot W_M$ given by
  $\dot W_M(A) (X,P) \teq ({W_M}X, \{ (x,B) \in {W_M} X ~|~x \in P, B \le
      A\}).
   $
    %
  Using the semantics of GHL with this \ghlstr{}, we develop a
  program logic for verification of dataflow-aware reuse bounds.

  \newcommand{\Passign}[1]{\mathsf{asgn}({#1})}

  Expressions with reuse bounds are given as procedures.  For each
  expression $e$, we define a function $\#e \colon \Var \to \NN$ that
  inductively counts variable usages.  We introduce a procedure name
  $\Passign e\in\PExp^s$ for each expression $e$ of sort $s$. Its
  interpretation $\sem{\Passign e} \colon \mem \to W_M(\mem)$ is
  $\sem{\Passign e} = (\sem{e}, A_{e})$ where
  $A_e(y,z) \teq \#e(y) \cdot \mathbf{d}_{x}(z) $.  We then define
  $\availproc s$ by
  \begin{align*}
    \availproc s(\phi, A, \psi)
    &=
      \{ \Passign{e} \mid
      e:s,
      \sem{\Passign{e}} \in \Pred_{\dot W_M}(\sem \phi,\sem \psi)(A_e),
      A_{e} \le A \}.
  \end{align*}
\end{example}

\end{toappendix}
\fi

\section{Introduction}
\label{sec:introduction}

The paradigm of \emph{grading} is
an emerging approach for
augmenting language semantics and type systems with fine-grained
information~\cite{DBLP:journals/pacmpl/OrchardLE19}. For example,
a \emph{graded monad} provides a mechanism for embedding
side-effects into a pure language, exactly as in the approach of
monads, but where the types are augmented (``graded'') with
information about what effects may occur, akin to a type-and-effect
system~\cite{DBLP:conf/popl/Katsumata14,DBLP:journals/corr/OrchardPM14}. As another example,
 \emph{graded comonadic} type operators in linear type systems can capture
non-linear dataflow and properties of data use~\cite{DBLP:conf/esop/BrunelGMZ14,DBLP:conf/esop/GhicaS14,DBLP:conf/icfp/PetricekOM14}. In
general, graded types augment a type system with some algebraic
structure which serves to give a parameterisable fine-grained program
analysis capturing the underlying structure of a type theory or semantics.
Much of the work in graded types has arisen in conjunction with
categorical semantics, in which graded modal type operators are
modelled via graded
monads~\cite{fujii2016towards,DBLP:conf/birthday/Gibbons16,DBLP:conf/fossacs/Katsumata18,DBLP:conf/birthday/MycroftOP16,milius2015generic},
graded comonads (often with additional graded monoidal
structure)~\cite{DBLP:conf/esop/BrunelGMZ14,DBLP:conf/esop/GhicaS14,DBLP:conf/fossacs/Katsumata18,DBLP:conf/icalp/PetricekOM13,DBLP:conf/icfp/PetricekOM14},
graded `joinads'~\cite{DBLP:conf/birthday/MycroftOP16}, graded
distributive laws between graded
(co)monads~\cite{DBLP:conf/icfp/GaboardiKOBU16},
and graded Lawvere theories~\cite{kura2020graded}.

So far grading has mainly been employed to reason about functional
languages and calculi, thus the structure of
the $\lambda$-calculus has dictated the structure of categorical
models (although some recent work connects graded
monads with classical dataflow analyses on
CFGs~\cite{ivakovi_et_al:LIPIcs:2020:12337}). We
investigate here the paradigm of grading instead applied to \emph{imperative}
languages.
As it happens, there is already a healthy thread of work in the
literature augmenting program logics (in the family of Floyd-Hoare
logics) with analyses that resemble notions of grading seen more
recently in the functional world. The general approach is to extend
the power of deductive verification by augmenting program logic rules
with an analysis of side effects, tracked by composing rules. For example,
work in the late 1980s and early 1990s augmented program logics with
an analysis of computation time, accumulating a cost
measure~\cite{nielson1987hoare,nielson1992semantics}, with more recent
fine-grained resource analysis based on multivariate analysis
associated to program
variables~\cite{DBLP:conf/pldi/Carbonneaux0S15}. As another example,
the Union Bound Logic of Barthe et al.~\cite{DBLP:conf/icalp/BartheGGHS16} defines
a Hoare-logic-style system for reasoning about probabilistic
computations with judgments $\vdash_\beta c : \phi \Rightarrow \psi$
for a program $c$ annotated by the maximum probability $\beta$
(the union bound) that $\psi$ does not hold. The inference rules of Union Bound
Logic track and compute the union bound alongside the standard
rules of Floyd-Hoare logic.  As a last example, Approximate Relational Hoare
Logic~\cite{2016arXiv160105047B,Barthe:2012:PRR:2103656.2103670,olmedo2014approximate,DBLP:journals/entcs/Sato16}
augments a program logic with measures of the $\epsilon$-$\delta$
bounds for reasoning about differential privacy.

In this work, we show how these disparate approaches can be unified
by adapting the notion of grading to an imperative program-logic
setting, for which we propose \emph{Graded Hoare Logic} (GHL): a
parameterisable program logic and reasoning framework graded by a
preordered monoidal analysis. Our core contribution is GHL's
underlying semantic framework which integrates grading, logical
assertions (pre- and post-conditions) and the effectful semantics of
an imperative language. This framework allows us to model, in a uniform way,
the different augmented program logics discussed above.

Graded models of functional calculi tend to adopt either a graded
monadic or graded comonadic model, depending on the direction of information flow
in the analysis. We use the opportunity of an imperative
setting (where the $\lambda$-calculus' asymmetrical
`many-inputs-to-one-output' model is avoided) to consider a more
flexible semantic basis of \emph{graded categories}. Graded categories
generalise graded (co)monadic approaches, providing a notion of graded
denotation without imposing on the placement (or `polarity') of grading.



\paragraph{Outline}
Section~\ref{sec:overview} begins with an overview of the approach, focusing on the
example of Union Bound Logic and highlighting the main components
of our semantic framework. The next three sections
then provide the central contributions:

\begin{itemize}
\item Section~\ref{sec:main} defines GHL and its associated
assertion logic which provides a flexible, parameterisable program
logic for integrating different notions of side-effect reasoning,
parameterised by a preordered monoidal analysis.
We instantiate the program logic to various examples.

\item Section~\ref{sec:graded-category} explores graded categories, an idea
that has not been explored much in the literature, and for which
there exists various related but not-quite-overlapping definitions. We
show that graded categories can abstract graded monadic and
graded comonadic semantics. We then extend graded categories
to Freyd categories (generally used as a more flexible model of
effects than monads), introducing the novel
structure of \emph{graded Freyd categories}.

\item Section~\ref{sec:model} develops the semantic framework for GHL,
 based on graded Freyd categories in a fibrational setting (where
\emph{coherent fibrations}~\cite{jacobscltt} model
the assertion logic) integrated with the graded Freyd layer.
We instantiate the semantic model to capture the examples presented in
Section~\ref{sec:main} and others drawn from the literature mentioned above.
\end{itemize}
%

\noindent
An extended version of this paper provides appendices which
include further examples and proof details~\cite{DBLP:journals/corr/abs-2007-11235}.

\section{Overview of GHL and Prospectus of its Model}
\label{sec:overview}

As discussed in the introduction, several works
explore Hoare logics combined with some form of
implicit or explicit grading for program analysis.
Our aim is to study these in a uniform way.
We informally introduce of our
approach here.

We start with an example which can be derived in Union Bound Logic~\cite{DBLP:conf/icalp/BartheGGHS16}:
\[
  \tripleV{0.05}{\top}{ \synProc{v_1}{\texttt{Gauss}(0,1);\synProc{v_2}{\texttt{Gauss}(0,1);\synExp v {\mathtt{max}(v_1,v_2)}}}}{v\leq 2}
\]
This judgment has several important components. First, we have
primitives for \emph{procedures with side-effects} such as
$\synProc{v_1}{\texttt{Gauss}(0,1)}$.
This procedure samples a random value from the standard normal
distribution with mean $0$ and variance $1$ and stores the result in
the variable $v_1$.
This kind of procedure with side effects differs from a regular
assignment such as $\synExp v {\mathtt{max}(v_1,v_2)},$ which is instead
considered to be pure (wrt. probabilities) in our approach.

The judgment has grade `$0.05$' which expresses a bound on the
probability that the postcondition is false, under the assumption of
the precondition, after executing the program; we can think of it as
the probability of failing to guarantee the postcondition.
In our example (call it program $P$), since the precondition is true, this can be expressed as:
  $
   \Pr_{\llbracket P\rrbracket(m)} [ v> 2]\leq 0.05
   $
   where $\llbracket P\rrbracket(m)$  is the probability distribution
   generated in executing the program. The grade of $P$ in this logic is derived using three components. First, sequential composition:
   $$
   \dfrac{\tripleV{\beta}{\psi}{P_1}{\psi_1}
      \quad \tripleV{\beta'}{\psi_1}{P_2}{\phi}}
      {\tripleV{\beta + \beta'}{\psi}{P_1; P_2}{\phi}}
      \quad
$$
which sums the failure probabilities. Second, an axiom for Gaussian distribution:
\[
  \tripleV{0.025}{\top}{ \synProc{v}{\texttt{Gauss}(0,1)}}{v\leq 2}
\]
with a basic constant $0.025$ which comes from the property of the Gaussian distribution we are considering. Third, by the following judgment which is derivable by the assignment and the consequence rules, which are the ones from Hoare Logic with a trivial grading $0$ which is the unit of addition:
\[
  \tripleV{0}{v_1 \leq 2 \lor v_2 \leq 2}{\synExp v {\mathtt{max}(v_1,v_2)}}{v\leq 2}
\]
Judgments for more complex examples can be derived using the rules for
conditional and loops. These rules also consider grading, and the
grading can depend on properties of the program. For example the rule
for conditionals is:
\begin{align*}
  \frac{
    \tripleV \beta{\psi\wedge e_b=\otrue}{P_1}\phi
    \quad
    \tripleV \beta{\psi\wedge e_b=\ofalse}{P_2}\phi
  }{
    \tripleV \beta\psi{\synIf{e_b}{P_1}{P_2}}\phi
  }
\end{align*}
This
allows one to reason also about the grading in a conditional way,
through the two assumptions $\psi\wedge e_b=\otrue$ and $\psi\wedge
e_b=\ofalse$.\marginpar{\dnote{I don't think this is clear}} We give more examples later.

Other logics share a similar structure as that described above for the
Union Bound logic, for example the relational logic
apRHL~\cite{2016arXiv160105047B}, and its variants
\cite{DBLP:journals/entcs/Sato16,DBLP:conf/lics/SatoBGHK19}, for
reasoning about differential privacy. Others again use a similar
structure implicitly, for example the Hoare Logic to reason about
asymptotic execution cost by Nielson~\cite{nielson1987hoare}, Quantitative
Hoare Logic~\cite{DBLP:conf/pldi/Carbonneaux0S15}, or the relational
logic for reasoning about program counter security presented
by Barthe~\cite{barthe20}.

To study the semantics of these logics in a uniform way, we
first abstract the logic itself. We design a program logic, which we
call Graded Hoare Logic (GHL), containing all the components
discussed above. In particular, the language is a standard imperative
language with conditional and loops. Since our main focus is studying
the semantics of grading, for simplicity we avoid using a
`while' loop, using instead a bounded `loop' operation
($\synLoop{e}{P}$). This allow us to focus on the grading structures for total functions, leaving the study of the interaction between grading and partiality to future work. The language
is parametric in the operations that are supported in
expressions---common in several treatments of Hoare Logic---and in a
set of procedures and commands with side effects, which are the main
focus of our work.  GHL is built over this language and an
\emph{assertion logic} which is parametric in the basic predicates
that can be used to reason about programs. GHL is also parametric in a
preordered monoid of grades, and in the axioms associated
with basic procedures and commands with side effects. This
generality is needed in order to capture the different logics we
mentioned before.

GHL gives us a unified syntax, but our real focus is the
semantics.  To be as general as possible we turn to the language
of category theory. We give a categorical framework which can capture
different computational models and side effects, with denotations
that are refined by predicates and grades describing program behaviours.
Our framework relates different categories (modelling different
aspects of GHL) as summarized by the following informal diagram
\eqref{eq:square}.
\begin{equation}
  \label{eq:square}
  \xymatrix@C=1.8em@R=1.8em{
    \PP \rdm{p} \rrrh{\dot I} & & \EE \rdh q \\
    \VV \rrrh{I} & & \CC
  }
\end{equation}
This diagram should not be understood as a
commutative diagram in $\CAT$ as $\EE$ is a graded category and hence
not an object of $\CAT$.

The category  $\VV$  models values
and pure computations, the category $\CC$ models impure computations, $\PP$ is a category of predicates, and
$\EE$ is a \emph{graded category} whose hom-sets are indexed by
\emph{grades}---elements of a preordered monoid.
The presentation of graded categories is new here, but has some relation
to other structures of the same name (discussed in Section~\ref{sec:graded-category}).

This diagram echos the principle of {\em refinement as functors}
proposed by Melli\`es and Zeilberger \cite{DBLP:conf/popl/MelliesZ15}.
The lower part of the diagram offers an interpretation of the language,
while the upper part offers a logical refinement of programs with
grading.  However, our focus is to introduce a new {\em graded refinement} view.
The ideas we use to achieve this are to interpret the base imperative language using a
{\em Freyd category} $I:\VV\arrow\CC$ (traditionally
used to model effects)
 with countable coproducts, to interpret the
assertion logic with a {\em coherent fibration} $p:\PP\arrow\VV$, and
to interpret GHL as a {\em graded Freyd category} $\dot I:\PP\arrow\EE$ with
homogeneous coproducts. In addition, the graded category $\EE$ has a
functor\footnote{More precisely,
  this is not quite a functor because $\EE$ is a graded category; see
  Definition \ref{def:logstr} for the precise meaning.} $q$ into $\CC$
which erases assertions and grades and extracts the denotation of
effectful programs, in the spirit of refinements.
The benefit of using a Freyd category as a building block
is that they are more flexible than other structures (e.g., monads) for
constructing models of computational
effects~\cite{10.1007/BFb0014560,STATON2014197}.  For instance, in the
category $\Meas$ of measurable spaces and measurable functions, we
cannot define state monads since there are no exponential objects.
However, we can still have a model of first-order effectful
computations using Freyd categories~\cite{power2006generic}.

Graded Freyd categories are a new categorical structure that we
designed for interpreting GHL judgments
(Section~\ref{sec:graded-freyd-categories}). The major difference from an
ordinary Freyd category is that the `target' category
is now a {\em graded category} ($\EE$ in the diagram \eqref{eq:square}). The additional structure
provides what we need in order to interpret judgments including grading.

To show the generality of this structure, we present several
approaches to instantiating the categorical framework of GHL's
semantics, showing constructions via graded monads and graded comonads
preserving coproducts.


Part of the challenge in designing a categorical semantics for GHL is
to carve out and implement the implicit assumptions and structures
used in the semantics of the various Hoare logics. A representative
example of this challenge is the interpretation of the rule for
conditionals in Union Bound Logic that we introduced above.
We interpret the assertion logic in (a variant of) coherent fibrations
$p:\PP\arrow\VV$, which model the
$\wedge${$\vee$}$\exists${${=}$}-fragment of first-order predicate
logic \cite{jacobscltt}. In this abstract setup, the rule for
conditionals may become {\em unsound} as it is built on the
implicit assumption that the type $\bool$, which is interpreted as
$1+1$, consists only of two elements, but this may fail in general $\VV$. For example, a suitable coherent fibration for
relational Hoare logic would take $\Set^2$ as the base category, but
we have $\Set^2(1,1+1)\cong 4$, meaning that there are four global elements in the
interpretation of $\bool$. We resolve this problem by introducing a
side condition to guarantee the decidability of the boolean
expression:
\begin{displaymath}
  \frac{
    \tripleV m{\psi\wedge e_b=\otrue}{P_1}\phi
    \quad
    \tripleV m{\psi\wedge e_b=\ofalse}{P_2}\phi
    \quad
    \textcolor{red}{\psi\vdash e_b=\otrue\vee e_b=\ofalse}
  }{
    \tripleV m\psi{\synIf{e_b}{P_1}{P_2}}\phi
  }
\end{displaymath}
This is related to the synchronization condition appearing in the
relational Hoare logic rule for conditional commands (e.g.,~\cite{Barthe:2012:PRR:2103656.2103670}).
%

Another challenge in the design of the GHL is how
to assign a grade to the loop command $\synLoop eP$. We may na\"{i}vely
give it the grade $m_l\teq\bigvee_{i\in\natrep{}} m^i$, where $m$ is the
grade of $P$, because $P$ is repeatedly executed some finite number of
times. However, the grade $m_l$ is a very loose over-approximation of
the grade of $\synLoop eP$. Even if we obtain some knowledge about the
iteration count $e$ in the assertion logic, this cannot be reflected
in the grade.  To overcome this problem, we introduce a Hoare logic
rule that can estimate a more precise grade of $\synLoop eP$, provided
that the value of $e$ is determined:
\begin{displaymath}
\dfrac{\forall 0\le z < N .\ \tripleV{m}{\psi_{z + 1}}{P}{\psi_z}\quad
  \psi_N\vdash e_n=\onat N
      }
      {\tripleV{m^N}{\psi_N}{\synLoop{e_n}{P}}{\psi_0}}
\end{displaymath}
This rule brings together the assertion language and grading,
creating a dependency from the former to the latter, and giving
us the structure needed for a categorical model.
The right premise is a judgment of the assertion logic (under
program variables $\mctx$ and pre-condition $\psi_N$)
requiring that $e$ is statically determinable as $N$. This premise makes the rule difficult to use in practical
applications where $e$ is dynamic.  We expect a more
``dependent'' version of this rule is possible with a more complex semantics
internalizing some form of data-dependency. Nevertheless, the above is enough to study the semantics of grading and its interaction with the Hoare Logic structure, which is our main goal here.




\section{Loop Language and Graded Hoare Logic}
\label{sec:main}
  \newcommand{\availcom}{C_{\mathsf c}}
  \newcommand{\availprocn}{C_{\mathsf p}}
  \newcommand{\availproc}[1]{\availprocn^{#1}}

After introducing some notation and basic concepts used throughout, we
outline a core imperative loop language, parametric in its set of
basic commands and procedures (Section~\ref{sec:syntax}).  We then
define a template of an assertion  logic (Section~\ref{sec:blogic}),
which is the basis of Graded Hoare Logic (Section~\ref{sec:ghl-definition}).

\subsection{Preliminaries}\label{sec:prelim}
Throughout, we fix an infinite set $\Var$ of variables
which are employed in \Lang{} (as the names of mutable program variables)
and in logic (to reason about these program variables).

A {\em many-sorted signature} $\Sigma$ is a tuple $(S,O,ar)$ where
$S,O$ are sets of sorts and operators, and $ar:O\arrow S^+$
assigns argument sorts and a return value sort to
operators (where $S^+$ is a non-empty sequence of sorts,
i.e., an operator $o$ with
signature $(s_1 \times \ldots \times s_n) \rightarrow s$ is
summarized as $\mathit{ar}(o) = \langle{s_1, \ldots, s_n, s\rangle} \in S^+$).
We say that another many-sorted signature $\Sigma'=(S',O',ar')$
is an {\em extension} of $\Sigma$ if $S\subseteq S'$ and $O\subseteq O'$
and $ar(o)=ar'(o)$ for all $o\in O$.

Let $\Sigma=(S,\cdots)$ be a many-sorted signature. A {\em context}
for $\Sigma$ is a (possibly empty) sequence of pairs
$\Gamma\in(\Var\times S)^*$ such that all variables in $\Gamma$ are
distinct. We regard $\Gamma$ as a partial mapping from
$\Var$ to $S$. The set of contexts for $\Sigma$ is denoted
$\Ctx\Sigma$. For $s\in S$ and $\Gamma\in\Ctx\Sigma$, we denote by
$\Exp \Sigma \Gamma s$ the set of $\Sigma$-expressions of sort $s$
under the context $\Gamma$. When $\Sigma,\Gamma$ are obvious,
we simply write $e:s$ to mean $e\in\Exp\Sigma\Gamma s$.
This set is inductively defined as usual.

An {\em interpretation} of a many-sorted signature $\Sigma=(S,O,ar)$
in a cartesian category $(\VV,1,\times)$ consists of an assignment of
an object $\sem s\in\VV$ for each sort $s\in S$ and an assignment
of a morphism
$\sem o\in\VV(\sem{s_1}\times\cdots\times \sem{s_n},\sem s)$ for each
$o\in O$ such that $ar(o)=\langle{s_1,\ldots,s_n,s}\rangle$.  Once such an
interpretation is given, we extend it to
$\Sigma$-expressions in the standard way (see, e.g.~\cite{crole1993categories,pitts1995categorical}).
First, for a context
$\Gamma=x_1:s_1,\cdots,x_n:s_n\in\Ctx\Sigma$, by $\sem\Gamma$ we mean
the product $\sem{s_1}\times\cdots\times\sem {s_n}$.  Then we
inductively define the interpretation of $e\in\Exp\Sigma\Gamma s$ as a
morphism $\sem e\in\VV(\sem\Gamma,\sem s)$.

\medskip

Throughout, we write bullet-pointed lists marked with
$\star$ for the mathematical data that are parameters to
Graded Hoare Logic (introduced in Section~\ref{sec:ghl-definition}).

\subsection{\LangTitleTitle{}}
\label{sec:syntax}

We introduce an imperative language called \Lang{}, with
a finite looping construct.  The language is parameterised by
the following data:
\begin{itemize}[leftmargin=1.2em]
\item[$\star$] a many-sorted signature $\Sigma=(S,O,ar)$ extending
  a base signature $(S_0,O_0,ar_0)$  of sort $S_0 = \{\sbool, \snat\}$ with
  essential constants as base operators $O_0$, shown here with their signatures
  for brevity rather than defining $\mathit{ar}_0$ directly:
  \begin{align*}
    O_0 = \{\otrue:\sbool,\ofalse:\sbool\} \cup \{\onat k:\snat~|~k\in\natrep{}\}
  \end{align*}
  where $\sbool$ is used for branching control-flow and $\snat$ is used for
  controlling loops, whose syntactic constructs are given below.
  We write $\onat k$ to mean the embedding of semantic
  natural numbers into the syntax.

\item[$\star$] a set $\CExp$ of {\em command names} (ranged over by $c$) and
  a set $\PExp_s$ of \emph{procedure names of sort $s$} (ranged over
  by $p$) for each sort $s\in S$.
\end{itemize}
When giving a program, we first fix a context $\mctx$ for the program variables. We
define the set of \emph{programs} (under a context $\mctx$) by the
following grammar:
\begin{align*}
  P & \,{::=}\, \synSeq{P}{P}
      \mid \synSkip{}
      \mid \synExp v {\exprc{}}
      \mid \synComm c
      \mid \synProc v p
      \mid \synIf{e_b}{P}{P}
      \mid \synLoop{e_n}{P}
\end{align*}
where $v\in\mctx$, $e_b,e_n$ are well-typed $\Sigma$-expressions of sort
$\sbool$ and $\snat$ under $\mctx$, and $c\in\CExp$. In assignment
commands, $\exprc{}\in\Exp\Sigma\mctx{\Gamma(v)}$. In procedure call
commands, $p\in \PExp_{\Gamma(v)}$. Each program must be
well-typed under $\mctx$. The typing rules are routine so we omit them.

Thus, programs can be sequentially composed via $;$ with $\synSkip$ as
the trivial program which acts as a unit to sequencing. An assignment
$\synExp v {\exprc{}}$ assigns expressions to a program variable $v$.
Commands can be executed through the instruction $\synComm c$ which
yields some side effects but does not return any value.  Procedures
can be executed through a similar instruction $\synProc v p$ which
yields some side effect but also returns a value which is used to update
$v$. Finally, conditionals are guarded by a boolean expression $e_b$
and the iterations of a looping construct are given by a natural
number expression $e_n$ (which is evaluated once at the beginning of
the loop to determine the number of iterations).

This language is rather standard, except for the treatment
of commands and procedures of which we give some examples here.
\begin{example}
  \label{ex:commands-procedures}
  {\em Cost Information:} a simple example of a command
  is $\tick{}$, which yields as a side effect the recording of one `step' of
  computation. 

  %
  {\em Control-Flow Information:} two other simple example of commands
  are $\cfTrue{}$ and $\cfFalse{}$, which yield as side effects the
  recording of either true or false to a log. A program can be
  augmented with these commands in its branches to give an account of
  a program's control flow.  We will use these commands to reason
  about control-flow security in Example
  \ref{exm-syn:program-counter}.

  {\em Probability Distributions:} 
  a simple example of a procedure is $\texttt{Gauss}(x,y)$, which yields as
  a side effect the introduction of new randomness in the program, and
  which returns a random sample from the Gaussian distribution with
  mean and variance specified by $x,y\in\mctx$. We will see how to use
  this procedure to reason about probability of failure in
  Example \ref{exm:syntax-union-bound-logic}.
  %
\end{example}



Concrete instances of \Lang{} typically include conversion
functions between the sorts in $\Sigma$, e.g., so that
programs can dynamically change control flow depending on values
of program variables.
In other instances, we may have a language
manipulating richer data types, e.g., reals or lists,
and also procedures capturing higher-complexity computations,
such as Ackermann functions.

\subsection{Assertion Logic}\label{sec:blogic}

We use an assertion logic to reason about properties
of basic expressions. We regard this reasoning as a meta-level
activity, thus the logic can have more sorts and operators
than \Lang{}. Thus, over the data specifying
\Lang{}, we build formulas of the assertion logic by the following
data:
\begin{itemize}[topsep=0.25em]
\item[$\star$] a many-sorted signature $\lsig=(S_l,O_l,ar_l)$
  extending $\Sigma$.
\item[$\star$] a set $P_l$ of atomic propositions and a function
  $par_l:P_l\arrow S_l^*$ assigning input sorts to them.
  We then inductively define the set $\Fml\lsig\Gamma$ of formulas
  under $\Gamma\in\Ctx\lsig$ as in Figure~\ref{fig:formula} (over the
  page), ranged over by $\psi$ and $\phi$.
\item[$\star$] a $\Ctx{\lsig}$-indexed family of subsets
  $\Axiom(\Gamma)\subseteq\Fml{\lsig}\Gamma\times\Fml{\lsig}\Gamma$.
\end{itemize}
The assertion logic is a fragment of the many-sorted first-order logic
over $\lsig$-terms admitting: 1) finite conjunctions, 2) countable
disjunctions, 3) existential quantification, and 4) equality
predicates. Judgements in the assertion logic have the form
$\Gamma~|~\psi_1,\cdots,\psi_n\vdash\phi$ (read as
$\psi_1\wedge\cdots\wedge\psi_n$ implies $\phi$), where
$\Gamma\in\Ctx{\lsig}$ is a context giving types to variables in the
formulas $\psi_1,\cdots,\psi_n,\phi\in\Fml{\lsig}\Gamma$. The logic
has the axiom rule deriving $\Gamma~|~\psi\vdash \phi$ for each pair
$(\psi,\phi)$ of formulas in $\Axiom(\Gamma)$.  The rest of inference rules of
this logic are fairly standard and so we omit them (see
e.g. \cite[Section 3.2 and Section 4.1]{jacobscltt}).

\begin{figure}[tp]
  The set
  $\Fml{\lsig}\Gamma$ of formulas under $\Gamma\in\Ctx{\lsig}$ is
  inductively defined as follows:
  \begin{enumerate}[topsep=0.25em]
  \item For all $p\in P_l$ and $par_l(p)=s_1\cdots s_n$ and
    $t_i:\Exp{\lsig}\Gamma{s_i}$ ($1\le i\le n$) implies
    $p(t_1,\cdots,t_n)\in\Fml{\lsig}\Gamma$
  \item For all $s\in S_l$ and $t,u\in\Exp{\lsig}\Gamma{s}$,
    $t=u\in\Fml{\lsig}\Gamma$.
  \item For all finite families
    $\{\phi_i\in\Fml{\lsig}\Gamma\}_{i\in\Lambda}$, we have
    $\bigwedge\phi_i\in\Fml{\lsig}\Gamma$.
  \item For all countable families
    $\{\phi_i\in\Fml{\lsig}\Gamma\}_{i\in\Lambda}$, we have
    $\bigvee\phi_i\in\Fml{\lsig}\Gamma$.
  \item For all $\phi\in\Fml{\lsig}{\Gamma,x:s}$, we have
    $(\ex{x:s}\phi)\in\Fml{\lsig}{\Gamma}$.
  \end{enumerate}
  \vspace{-0.5\baselineskip}
  \caption{Formula formation rules}
  \label{fig:formula}
  \vspace{-\baselineskip}
\end{figure}


In some of our examples we will use the assertion logic to reason
about programs in a relational way, i.e., to reason about two
executions of a program (we call them {\em left} and {\em right}
executions).  This requires basic predicates to manage expressions representing
pairs of values in our assertion logic. As an example, we could have two predicates
$\mathsf{eqv}_{\langle 1\rangle}$, $\mathsf{eqv}_{\langle 2\rangle}$,
that can assert the equality of the left and right executions of an
expression to some value, respectively. That is, the formula
$\mathsf{eqv}_{\langle 1\rangle} (e_b,{\sf true})$, which we will
write using infix notation $ e_b\langle 1\rangle = {\sf true}$,
asserts that the left execution 
of the boolean expression $e_b$ is equal to ${\sf true}$.


\subsection{Graded Hoare Logic}
\label{sec:ghl-definition}

We now introduce Graded Hoare Logic (GHL), specified by the following
data:
\begin{itemize}
\item[$\star$] a preordered monoid $(M,\le,1,\cdot)$ (\emph{pomonoid} for
  short) (where $\cdot$ is monotonic with respect to $\le$) for the purposes of
  program analysis, where we refer to the elements $m \in M$ as \emph{grades};
\item[$\star$] two functions which define the grades and
  pre- and post-conditions of commands $\CExp$ and
  procedures $\PExp$:
  \begin{align*}
    \availcom &:\Fml{\lsig}\mctx\times M\arrow 2^\CExp\\
    \availproc s &:\Fml{\lsig}\mctx\times M\times
                   \Fml{\lsig}{r:s}\arrow 2^{\PExp_s}\quad (s\in S \wedge r
                   \not\in \mathsf{dom}(\mctx))
  \end{align*}
  %
\end{itemize}
The function $\availcom$ takes a pre-condition and a grade, returning
a set of command symbols satisfying these specifications. A command
$c$ may appear in $\availcom(\phi,m)$ for different pairs $(\phi,m)$,
enabling pre-condition-dependent grades to be assigned to $c$.
Similarly, the function $\availproc s$ takes a pre-condition, a grade,
and a postcondition for return values, and returns a set of
procedure names of sort $s$ satisfying these specifications. Note,
$r$ is a distinguished variable (for return values) not in $\mctx$.
The shape of $\availcom$ and $\availcom$ as predicates over commands
and procedures, indexed by assertions and grades, provides a way
to link grades and assertions for the effectful operations of
GHL. Section~\ref{sec:examples} gives examples exploiting this.

From this structure we define a \emph{graded Hoare logic}
by judgments of the form:
$
  \tripleV{m}{\phi}{P}{\psi}
$
denoting a program $P$ with pre-condition $\phi\in \Fml{\lsig}\mctx$,
postcondition $\psi\in \Fml{\lsig}\mctx$ and analysis
$m \in M$. Graded judgments are defined inductively via the inference
rules given in Table~\ref{fig:rules}.
\begin{table}[tbp]
  \begin{align*}
    \begin{array}{c}
      %
      %
      \dfrac{}{\tripleV{1}{\psi}{\synSkip}{\psi}}
      \quad
      %
      %
      \dfrac{\tripleV{m}{\psi}{P_1}{\psi_1}
      \quad \tripleV{m'}{\psi_1}{P_2}{\phi}}
      {\tripleV{m \cdot m'}{\psi}{P_1; P_2}{\phi}}
      \quad
      %
      %
      \dfrac{}{\tripleV{1}{\psi\subst{e}{v}}{v := e}{\psi}}
      \\[2em]
      %
      %
      \dfrac{f\in \availcom(\psi,m)}
      {\tripleV{m}{\psi}{\synComm{c}}{\psi}}\quad
      \dfrac{p\in \availproc{\mctx(v)}(\psi,m,\phi)}
      {\tripleV{m}{\psi}{\synProc{v}{p}}{(\ex{v:\mctx(v)}\psi)\wedge\phi[v/r]}}
      \\[2em]
      %
      %
      \dfrac{
      \mctx~|~\psi' \vdash \psi \quad 
      m \leq m' \quad
      \mctx~|~\phi \vdash \phi' \quad
      \tripleV{m}{\psi}{P}{\phi}}
      {\tripleV{m'}{\psi'}{P}{\phi'}}
      \\[2em]
      %
      %
      %
      %
      %
      \dfrac{\forall 0\le z < N .\ \tripleV{m}{\psi_{z + 1}}{P}{\psi_z}\quad
      \mctx~|~\psi_N\vdash e_n=\onat N
      }
      {\tripleV{m^N}{\psi_N}{\synLoop{e_n}{P}}{\psi_0}}
      \\[2em]
      %
      %
      \dfrac{
      \begin{array}{l@{}}
      \tripleV{m}{\psi \wedge e_b = \otrue}{P_1}{\phi}
      \quad \tripleV{m}{\psi \wedge e_b = \ofalse}{P_2}{\phi}
      \quad \mctx~|~\psi\vdash e_b=\otrue\vee e_b=\ofalse
      \end{array}
      }
      {\tripleV{m}
      {\psi}{\synIf{e_b}{P_1}{P_2}}{\phi}}
      \\[1em]
      %
      %
      %
    \end{array}
  \end{align*}
  \caption{Graded Hoare Logic Inference Rules}
  \label{fig:rules}
  \vspace{-\baselineskip}
\end{table}
Ignoring grading, many of the rules are fairly standard for a
Floyd-Hoare program logic.  The rule for \synSkip{} is standard but
includes grading by the unit $1$ of the monoid
. Similarly, assignment is standard, but graded with
$1$ since we do not treat it specially in GHL. Sequential composition
takes the monoid multiplication of the grades of the subterms.  The
rules for commands and procedures use the functions $\availcom$ and $\availprocn$
introduced above.  Notice that the rule for commands uses as
the pre-condition as its post-condition, since commands have only side
effects and they do not return any value.
The rule for procedures combines the pre- and post-conditions given by $\availprocn$
following the style of Floyd's assignment rule~\cite{Floyd1967}.

The non-syntax-directed consequence rule is similar to the usual consequence rule, and
in addition allows the assumption on the grade to be weakened (\emph{approximated})
according to the ordering of the monoid.

The shape of the loop rule is slightly different from the usual one.
It uses the assertion-logic judgment $\mctx~|~\psi_N\vdash e_n=\onat N$ to express the
assumption  that $e_n$ evaluates to $\onat N$. Under this assumption
it uses a family of assertions $\psi_z$ indexed by the natural numbers
$z \in \{0,1,\ldots, N-1\}$
to conclude the post-condition $\psi_0$. This family of assertions plays the role of
the classical invariant in the Floyd-Hoare logic rule for `while'.
Assuming that the grade of the loop body is $m$, the grade of the loop command
is then $m^N$, where $m^0 = 1$ and $m^{k+1} = m \cdot m^k$.
By instantiating this rule with $\psi_z = (\theta \land e_n = {\onat z})$, the
loop rule also supports the following derived rule which is often preferable in examples:
\[
  \dfrac{\forall 0\le z < N .\ \tripleV{m}{\theta \land e_n = {\onat
        {z + 1}} }{P}{\theta \land e_n = {\onat z}
    } 
  } {\tripleV{m^N}{\theta \land e_n = {\onat
        N}}{\synLoop{e_n}{P}}{\theta \land e_n = {\onat 0}}}
\]
The rule for the conditional is standard except for the condition
$\mctx~|~\psi \vdash e_b=\otrue\vee e_b=\ofalse$. While this condition
may seem obvious, it is actually important to make 
GHL sound in various semantics (mentioned in Section~\ref{sec:overview}).
%
{As an example,
  suppose that a semantics $\sem\blank$ of expressions is given in the
  product category $\Set^2$, which corresponds to two semantics
  $\sem\blank_1,\sem\blank_2$ of expressions in $\Set$. Then the side
  condition for the conditional is to guarantee that for any boolean
  expression $e_b$, and pair of memories $(\rho_1,\rho_2)$ satisfying
  the precondition $\psi$, the pair
  $(\sem{e_b}_1(\rho_1),\sem{e_b}_2(\rho_2))$ is either
  $\sem{\otrue}=(\mtrue,\mtrue)$ or $\sem{\ofalse}=(\mfalse,\mfalse)$.
  We note that other relational logics such as
  apRHL~\cite{Barthe:2012:PRR:2103656.2103670} employ an
  equivalent syntactic side condition in their rule for conditionals.  }
%

\subsection{Example Instantiations of GHL}
\label{sec:examples}
\begin{example}[Simple cost analysis]
  \label{exm-syn:simple-cost}
  We can use the $\tick{}$ command discussed in
  Example~\ref{ex:commands-procedures} to instrument programs with
  \emph{cost} annotations. We can then use GHL to perform cost
  analysis by instantiating GHL with the additive natural
  number monoid $(\NN, \leq, 0, +)$ and
  $\tick{} \in \availcom(\phi, 1)$. Thus, we can form
  judgments $\tripleV{1}{\phi}{\synComm{\tick}}{\phi}$
  which account for cost via the judgment's grade. Sequential
  composition accumulates cost and terms like $\synSkip$ and
  assignment have $0$ cost.

  Let us use this example to illustrate how $\availcom$ can assign
  multiple pre-condition-grade pairs to a command. Suppose that we
  modify the semantics of $\tick{}$ so that it reports unit cost $1$ when
  variable $x$ is $0$, otherwise cost $2$. We can then define
  $\availcom$ so that $\tick{}\in \availcom(x=\onat 0,1)$ and also
  $\tick{}\in\availcom(x\neq\onat 0,2)$.  In this way,
  we can give different grades to
  programs depending  on their pre-conditions.
\end{example}

\begin{example}[Program Counter Security]
  \label{exm-syn:program-counter}
  We can use the commands $\cfTrue$ and $\cfFalse$ discussed in
  Example~\ref{ex:commands-procedures} to instrument programs with
  \emph{control flow} annotations, recording to an external log. GHL can then be used to reason
  about program counter security~\cite{MolnarPSW05}\cite[Section 7.2]{barthe20} of
  instrumented programs. This is a relational security property
  similar to non-interference (requiring that private values do not
  influence public outputs) but where only programs with the same
  control flow are considered.

  Firstly, any conditional
  statement $\synIf{e_b}{P_t}{P_f}$ in a program is
  elaborated to a statement
  $\synIf{e_b}{(\cfTrue; P_t)}{(\cfFalse; P_f)}$.  We then
  instantiate GHL with a monoid of words over $\{\otrue,\ofalse\}$ with prefix
  order: $2^*\teq(\{\otrue,\ofalse\}^*, \leq, \epsilon, \cdot)$ and we consider
  $\cfTrue \in \availcom(\phi, \otrue)$ and $\cfTrue \in
  \availcom(\phi, \ofalse)$.
  We can thus form judgments of the shape
  $\tripleV{\otrue}{\phi}{\synComm{\cfTrue}}{\phi}$ and
  $\tripleV{\ofalse}{\phi}{\synComm{\cfFalse}}{\phi}$ which account for
  control-flow information (forming paths) via the judgment's
  grade. Sequential composition concatenates control-flow paths and
  terms like $\synSkip$ and assignment do not provide any control-flow information, i.e. $\epsilon$.

  We then instantiate the assertion logic to support
  relational reasoning, i.e., where the expressions of the language
  are interpreted as pair of values. For an expression $e$, interpreted as a
  pair $(v_1, v_2)$ then we write $e\langle{1}\rangle = v_1$ to say
  that the first component (left execution) equals $v_1$ and $e\langle{2}\rangle
  = v_2$ to say that the second component (right execution) equals $v_2$. In the
  assertion logic, we can then describe public values which need to be
  equal, following the tradition in reasoning about non-interference,
  by the predicate $e\langle{1}\rangle = e\langle{2}\rangle$. Private data
  are instead interpreted as a pair of arbitrary values.
  (Section~\ref{sec:blogic} suggested the notation
  $\mathsf{eqv}_{\langle i\rangle}(e,b)$ for $e\langle{i}\rangle = b$,
  but we use the latter for compactness here).


  As an example,
  one can prove the following judgment where $x$ is a public variable
  and $y$ is a private one, and $b\in\{\otrue,\ofalse\}$:
%
\marginpar{\tnote{[review] This does not actually use the eqPub notation it is
   supposed to illustrate.}\tnote{I do not understand the meaning of
   this reviewer's comment.}
\dnote{We are mixing notations (eqPub above then eqv here).}}
  {\small $$\tripleV{b}{ x\langle 1\rangle{=}x\langle 2\rangle\wedge  x\langle 1\rangle{=}b}{\synIf{x}{(\cfTrue; x{=}1;y{=}1)}{(\cfFalse; x{=}2;y{=}2)}}{x\langle 1\rangle{=}x\langle 2\rangle}$$}
  This judgment shows the program is non-interferent, since the value
  of $x$ is independent from the value of the private variable $y$,
  and secure in the program counter model, since the control flow does
  not depend on the value of $y$. Conversely, the following judgment
  is not derivable for both $b=\otrue$ and $b=\ofalse$:
  {\small $$\tripleV{b}{  x\langle 1\rangle{=}x\langle 2\rangle\wedge  y\langle 1\rangle{=}b}{\synIf{y}{(\cfTrue; x{=}1;y{=}1)}{(\cfFalse; x{=}1;y{=}2)}}{x\langle 1\rangle{=}x\langle 2\rangle}$$}
This program is non-interferent but is not secure in the program counter model because the control flow leaks information about $y$ which is a private variable.
\end{example}

\newcommand{\countVar}[1]{\mathsf{count}(#1)}

\begin{example}[Union Bound Logic]
\label{exm:syntax-union-bound-logic} 
Section~\ref{sec:introduction} discussed the Union Bound logic by Barthe et
al.~\cite{DBLP:conf/icalp/BartheGGHS16}.  This logic embeds smoothly
into GHL by using the \pmonoid{}
$(\mathbb{R}_{\geq 0}, \leq, 0,+)$ and procedures of the form
$\osample\mu\exprc$ as samplings from a probabilistic
distribution $\mu$ parametrised over the syntax of
GHL expressions $\exprc$. Following Barthe et
al.~\cite{DBLP:conf/icalp/BartheGGHS16}, we consider a
semantically defined set for $C_p$:
\begin{align*}
C_p(\phi, \beta, \psi)
 = \{ \osample\mu\exprc \mid
    \forall s . s \in \interp{\phi}{\implies}
    \text{Pr}_{s'\leftarrow\llbracket\osample\mu\exprc\rrbracket(s)}[s'\in\interp{\neg \psi}] \leq \beta) \}
\end{align*}
This definition captures that, assuming the pre-condition holds for
an input memory state $s$, then for output
value $s'$ from sampling $\osample\mu\exprc$,
the probability that the post-condition is false is bounded above by
$\beta$.  This allow us to consider different properties of the
distribution $\mu$ with parameter $e$.
\end{example}




\section{Graded Categories}
\label{sec:graded-category}

Now that we have introduced GHL and key examples, we turn to the core of
its categorical semantics: \emph{graded categories}.

Graded monads provide a notion of sequential composition for morphisms
of the form $I \rightarrow T_m J$, i.e., with structure on the
target/output capturing some information by the grade $m$ drawn
from a \pmonoid{}~\cite{DBLP:conf/popl/Katsumata14}; dually,
graded comonads provide composition for $D_m I \rightarrow J$,
i.e. with
 structure on the source/input with grade $m$~\cite{DBLP:conf/icalp/PetricekOM13}. We avoid the choice
of whether to associate grading with the input or output by
instead introducing \emph{graded categories}, which are
agnostic about the polarity (or position) of any structure and grading.
Throughout this section, we fix a \pmonoid{} $(M,\le,1,\cdot)$
(with $\cdot$ monotonic wrt. $\le$).
\newcommand{\up}[4]{{\uparrow}_{#3}^{#4}}
\begin{definition}\conf{90}\label{def:grcat}
  An {\em $M$-graded category} $\CC$ consists of the following data:
  \begin{itemize}[leftmargin=1.5em]
  \item A class $\Obj\CC$ of objects. $I\in\CC$ denotes $I\in\Obj\CC$.
  \item A homset $\CC(I,J)(m)$ for all objects $I,J\in\CC$ and
    $m\in M$. We often write $f:I\garrow{m}J$ to mean
    $f\in\CC(I,J)(m)$, and call $m$ the {\em grade} of
    $f$;
   %
  \item An upcast functions
    $\up I J m n:\CC(I,J)(m)\arrow\CC(I,J)(n)$ for all grades $m\le n$;
  %
  \item Identity morphisms $\id_I\in\CC(I,I)(1)$ for all $I\in\CC$;
  \item Composition
    $\circ:\CC(J,K)(n)\times\CC(I,J)(m)\arrow\CC(I,K)(m\cdot n)$.
  \end{itemize}
  Graded categories satisfy the usual categorical laws of identity and
  associativity, and also the commutativity of upcast and composition:
  $\up J K n {n'} g
    \circ
    \up I J m  {m'} f=
    \up I K {m\cdot n}{m'\cdot n'}(g\circ f)$,
  corresponding to monotonicity of $(\cdot)$ with respect to $\le$.
\end{definition}
An intuitive meaning of a graded category's morphisms is:
$f \in \CC (A, B) (m)$ if the {\em{value}} or the {\em{price}} of a
morphism $f : A \rightarrow B$ is {\em{at most}} $m$ with respect to
the ordering $\le$ on $M$.
%
%
%
We do not yet give a polarity or direction to this price, i.e.,
whether the price is {\em{consumed}} or {\em{produced}} by the
computation. Thus, graded categories give a non-biased view; we need
not specify whether grading relates to the source or target of a
morphism.

Graded categories were first introduced by Wood
\cite[Section 1]{Wood} (under the name `large $V$-categories'), and Levy connected them with models of
call-by-push-value \cite{locgracat}.  Therefore we do not claim the
novelty of Definition \ref{def:grcat}.


\newcommand{\Tup}[3]{{#1}({#2}\le{#3})}
\begin{example}\conf{80}\label{ex:KleisliFreyd}
  A major source of graded categories is via graded (co)monads.  Let
  $(M,\le,1,\cdot)$ be a \pmonoid{}, regarded as a monoidal
  category. A {\em graded monad}
  \cite{smirnov2008graded,DBLP:conf/popl/Katsumata14} on a category
  $\CC$ (or more precisely an \emph{$M$-graded monad}) is a lax
  monoidal functor
  $(T,\eta,\mu):(M,\le,1,\cdot)\arrow ([\CC,\CC],\Id,\circ)$.
  Concretely, this specifies:
  \begin{itemize}
  \item a functor $T:(M,\le)\arrow [\CC,\CC]$ from the preordered set
    $(M,\le)$ to the endofunctor category over $\CC$. For an ordered
    pair $m\le m'$ in $M$ then $\Tup T m {m'} : Tm\arrow Tm'$ is a natural
    transformation;
  \item a unit $\eta:\Id\arrow T1$ and a multiplication
    $\mu_{m,m'}:Tm\circ Tm'\arrow T(m\cdot m')$, natural in
    $m,m'\in M$.
  \end{itemize}
  They satisfy the graded versions of the usual monad axioms:
   \begin{displaymath}
     \hspace{-0.5em}
     \xymatrix@C=1.8em@R=1.8em{
       TmJ \rrh{Tm\eta_J} \rrde \rdm{\eta_{TmJ}} & Tm(T1J) \rdh{\mu_{m,1,J}} &
       Tm(Tm'(Tm''J)) \rrrh{\mu_{m,m',Tm''J}} \rdm{Tm\mu_{m',m'',J}} &
       & T(m \cdot m')(Tm''J) \rdh{\mu_{mm',m'',J}} \\
       T1(TmJ) \rrh{\mu_{1,m,J}} & TmJ &
       Tm(T(m' \cdot m'')J) \rrrh{\mu_{m,m'm'',J}} & & T(m \cdot m' \cdot m'')J
     }
   \end{displaymath}
  Graded comonads are dually defined (i.e., as a graded monad on
  $\CC^{op}$).

  By mimicking the construction of Kleisli categories, we can
  construct an $M$-graded category $\CC_T$ (we call it the Kleisli
  $M$-graded category of $T$) from a category $\CC$ with an $M$-graded
  monad $T$ on $\CC$.\footnote{Not to be confused with the Kleisli
    category of graded monads by Fujii et
    al.~\cite{fujii2016towards}.}
  \begin{itemize}
  \item  $\Obj{\CC_T}\teq\Obj\CC$
    and $\CC_T(X,Y)(m)\teq \CC(X,TmY)$.
  \item For $f:X\garrow{m}Y$ and $n$ such that $m \leq n$, we define
    $\up X Y m n f \teq \Tup T m n_Y\circ f$.
  \item Identity and composition are defined by:
    $\id_X\teq \eta_X:X\garrow 1X$ and
    $g\circ f\teq\mu_{m,n,Z}\circ T  m\, g \circ f$ for $f:X\garrow{m}Y$ and
    $g:Y\garrow{n}Z$.
  \end{itemize}
  The dual construction is possible. Let $D$ be an $M^\op$-graded
  comonad on a category $\CC$.  We then define $\CC_D$ by
  $\CC_D(X,Y)(m)=\CC(DmX,Y)$; the rest of data is similar to the case
  of graded monads. This yields an $M$-graded category $\CC_D$.
\end{example}

\begin{remark}
  As an aside (included for completeness but not needed in the rest of
  the paper), graded categories are an instance of {\em enriched
    categories}. For the enriching category, we take the presheaf
  category $[M,\Set]$, together with {\em Day's convolution product}
  \cite{Day}.

\end{remark}
%

\subsection{Homogeneous Coproducts in Graded Categories \conf{80}}

We model boolean values and natural numbers by the binary coproduct
$1+1$ and the countable coproduct $\coprod_{i\in\natrep}1$.  We thus
define what it means for a graded category to have coproducts. The
following definition of binary coproducts easily extends to coproducts
of families of objects.
\begin{definition}
  Let $\CC$ be an $M$-graded category. A {\em{homogeneous binary
      coproduct}} of $X_1, X_2 \in \CC$ consists of an object
  $Z\in \CC$ together with injections $\iota_1\in\CC(X_1,Z)(1)$ and
  $\iota_2\in\CC(X_2,Z)(1)$ such that, for any $m\in M$ and $Y\in\CC$,
  the function $ \lam{f}(f\circ\iota_1,f\circ\iota_2)$ of type
  $\CC(Z,Y)(m)\arrow\CC(X_1,Y)(m)\times\CC(X_2,Y)(m)$ is
  invertible. The inverse is called the {\em cotupling} and denoted by
  $[-,-]$.  It satisfies the usual law of coproducts ($i=1,2$):
  \begin{align*}
    [f_1,f_2]\circ\iota_i&=f_i,\quad
    &[\iota_1,\iota_2]&=\id_{Z},\\
    g\circ [f_1,f_2]&= [g\circ f_1,g\circ f_2],\quad
    &[\up {X_1}Ymnf_1,\up {X_2}Ymnf_2]&=\up{Z}Ymn[f_1,f_2].
  \end{align*}
  When homogeneous binary coproducts of any combination of
  $X_1,X_2\in\CC$ exists, we say that $\CC$ has homogeneous binary
  coproducts.
\end{definition}
The difference between homogeneous coproducts and coproducts
 in ordinary category theory is that the
cotupling is restricted to take morphisms with the same grade. A
similar constraint is seen in some effect systems, where the typing
rule of conditional expressions require each branch to have the
same effect.
\begin{proposition}
  \label{pp:Gradedmonad:coproduct}
  Let $\{\iota_i\in\CC(X_i,Z)\}_{i\in I}$ be a
  coproduct of $\{X_i\}_{i\in I}$ in an ordinary category $\CC$.
  \begin{enumerate}
  \item Suppose that $T$ is an $M$-graded monad on $\CC$. Then
    $\{\eta_Z\circ\iota_i\in\CC_T(X_i,Z)(1)\}_{i\in I}$ is a
    homogeneous coproduct in $\CC_T$.
  \item Suppose that $(D,\varepsilon,\delta)$ is an $M^{op}$-graded
    comonad on $\CC$ such that each $Dm:\CC\arrow\CC$ preserves the
    coproduct $\{\iota_i\}_{i\in I}$.  Then
    $\{\iota_i\circ \varepsilon_I\in\CC_D(X_i,Z)(1)\}_{i\in I}$ is a
    homogeneous coproduct in $\CC_D$.
  \end{enumerate}
\end{proposition}
\begin{appendixproof}
  (Proof of Proposition \ref{pp:Gradedmonad:coproduct}) (1) Let
  $Y\in\CC_T$ and $f_i:X_i\garrow mY$ be $I$-indexed morphisms.  Then
  the cotupling $[f_i]_{i\in I}$ taken in $\CC$ is a morphism of type
  $Z\garrow mY$ in $\CC_T$. It is easy to verify that
  $\{f_i\}_{i\in I}\mapsto [f_i]_{i\in I}$ is the inverse of
  $\lam{h}(\mu_{1,m,Z}\circ T1h\circ \eta_Z\circ \iota_i)_{i\in I}$.

  (2) This case can be proven similarly, using the fact that $Dm$
  preserves coproducts.
\end{appendixproof}

\subsection{Graded Freyd Categories with Countable Coproducts}
\label{sec:graded-freyd-categories}

We now introduce the central categorical structure of
\Lang{} and GHL semantics: {\em graded
  Freyd categories} with homogeneous countable coproducts.
%
\begin{definition}\label{def:freydcoprod}
  An $M$-graded Freyd category with homogeneous countable coproducts
  consists of the following data:
  \begin{enumerate}
  \item A cartesian monoidal category $(\VV,1,\times,l,r,a)$ with
    countable coproducts such that for all $V \in \VV$,
    the functor
    $V\times(\blank):\VV\arrow\VV$
    preserves coproducts.
  \item An $M$-graded category $\CC$ such that $\Obj\CC=\Obj\VV$ and $\CC$
    has homogeneous countable coproducts.
  \item A function $I_{V,W}:\VV(V,W)\arrow\CC(V,W)(1)$ for each
    $V,W\in\CC$. Below we may omit writing subscripts of $I$.  The
    role of this function is to inject pure computations into effectful
    computations.
  \item A function
    $(*)_{V,X,W,Y}:\VV(V,W)\times\CC(X,Y)(m)\arrow\CC(V\times
    X,W\times Y)(m)$ for each $V,W,X,Y\in\CC$ and $m\in M$. Below we
    use it as an infix operator and sometimes omit its subscripts.
    The role of this function is to combine pure computations and
    effectful computations in parallel.
  \end{enumerate}
  The function $I$ and $(*)$ satisfy the following
  equations:
  \begin{gather*}
    I(\id_X)=\id_X\quad I(g\circ f)=Ig\circ If \quad I(f\times g)=f*Ig
    \quad \id_V*\id_X=\id_{V*X},
    \\
    (g\circ f)*(i\circ j)=(g*i)\circ (f*j) \quad f*\up X Y m n g=\up
    {V*X} {W*Y} m n (f*g)
    \\
    f\circ I(l_X) = I(l_X)\circ(\id_1*f) \quad
    I(a_{X'\!,Y'\!,Z'\!})\circ ((f{\times} g)*h) =(f*(g{*}h)) \circ
    I(a_{X,Y,Z})
  \end{gather*}
  These are analogous to the usual Freyd categories axioms.
  We also require that:
  \begin{enumerate}
  \item \label{cond:icoprod}For any countable coproduct
    $\{\iota_i\in\VV(X_i,Y)\}_{i\in A}$,
    $\{I(\iota_i)\in\CC(X_i,Y)(1)\}_{i\in A}$ is a homogeneous
    countable coproduct.
  \item For any homogeneous countable coproduct
    $\{\iota_i\in\CC(X_i,Y)(1)\}_{i\in A}$ and $V\in\VV$,
    $\{\id_V*\iota_i\in\CC(V\times X_i,V\times Y)(1)\}_{i\in A}$ is a
    homogeneous countable coproduct.
  \end{enumerate}
  We denote an $M$-graded Freyd category with countable coproducts by
  the tuple $(\mathbb{V}, 1, \times, \mathbb{C}, I, (\ast))$ capturing
  the main details of the cartesian monoidal structure of
  $\mathbb{V}$, the base category $\mathbb{C}$, the lifting function
  $I$ and the action $(*)$.
\end{definition}
\noindent
If the grading \pmonoid{} $M$ is trivial, $\CC$ becomes an ordinary
category with countable coproducts. We therefore simply call it a
Freyd category with countable coproducts. This is the same as a {\em
  distributive Freyd category} in the sense introduced by
Power~\cite{power2006generic} and
Staton~\cite{STATON2014197}. \draft[A new ``example'' we promised
  during rebuttal]{We will use non-graded Freyd categories to give a
  semantics of \Lang{} in Section~\ref{sec:semloop}. An advantage of Freyd categories is that they
  encompasses a broad class of models of computations, not limited to
  those arising from monads. A recent such example is Staton's
  category of {\em s-finite kernels}
  \cite{DBLP:conf/esop/Staton17}\footnote{It is not known whether the
    category of s-finite kernels is a Kleisli category.}.}

We could give an alternative abstract definition of $M$-graded
Freyd category using $2$-categorical language: a graded Freyd category
is an equivariant morphism in the category of actions from a cartesian
category to $M$-graded categories.
The full detail of this formulation will be discussed elsewhere.

A Freyd category typically arises from a strong monad on a cartesian
category \cite{10.1007/BFb0014560}. We give here a graded analogue of
this fact. First, we recall the notion of {\em strength} for graded
monads \cite[Definition 2.5]{DBLP:conf/popl/Katsumata14}.  Let
$(\CC,1,\times)$ be a cartesian monoidal category.  A {\em strong}
$M$-graded monad is a pair of an $M$-graded monad $(T,\eta,\mu)$ and a
natural transformation
$\mathrm{st}_{I,J,m}\in\CC(I\times TmJ,Tm(I\times J))$ satisfying
graded versions of the four coherence laws in \cite[Definition
3.2]{DBLP:journals/iandc/Moggi91}. We dually define a {\em
  costrong} $M$-graded comonad $(D,\varepsilon,\delta,\mathrm{cs})$ to be
the $M$-graded comonad equipped with the {\em costrength}
$\mathrm{cs}_{I,J,m}\in\CC(Dm(I\times J),I\times DmJ)$.
\begin{toappendix}
  \section{Coherence Laws of Strength}
  \begin{displaymath}
    \xymatrix{
      1\times TmX \rr \rrd & Tm(1\times X) \rd\\
      & TmX
    }
    \xymatrix{
      X\times (Y\times TmZ) \rr \rd & X\times Tm(Y\times Z) \rr & Tm(X\times (Y\times Z)) \rd \\
      (X\times Y)\times TmZ \rrr & & Tm((X\times Y)\times Z)
    }
  \end{displaymath}
  \begin{displaymath}
    \xymatrix{
      X\times Y \rr \rrd & X\times T1Y \rd\\
      & T1(X\times Y)
    }
    \xymatrix{
      X\times Tm(Tm'Y) \rr \rd & Tm(X\times Tm'Y) \rr & Tm(Tm'(X\times Y)) \rd \\
      X\times T(mm')Y \rrr & & T(mm')(X\times Y)
    }
  \end{displaymath}
\end{toappendix}
\begin{proposition}\label{pp:Klieslifreyd}\label{pp:Klieslifreyd:coprod}
  Let $(\CC,1,\times)$ be a cartesian monoidal category.
  \begin{enumerate}
  \item Let $(T,\eta,\mu,\mathrm{st})$ be a strong $M$-graded monad on
    $\CC$. The Kleisli $M$-graded category $\CC_T$, together with
    $If=\eta_W\circ f$ and
    $f\ast g=\mathrm{st}_{W,Y} \circ (f \times g)$ forms an $M$-graded
    Freyd category with homogeneous countable
    coproducts.

  \item Let $(D,\varepsilon,\delta,\mathrm{cs})$ be a costrong
    $M^{op}$-graded comonad on $\CC$ such that each $Dm$ preserves
    countable coproducts. Then the coKleisli $M$-graded category
    $\CC_D$ together with $If=f\circ\varepsilon_V$ and
    $f*g=(f\times g)\circ\mathrm{cs}_{V,X}$ forms an $M$-graded Freyd
    category with homogeneous countable coproducts.
  \end{enumerate}
\end{proposition}
\begin{appendixproof}
  (Proof of Proposition \ref{pp:Klieslifreyd}) Thanks to the strong
  graded monad structure of $T$, we easily satisfy the equalities for
  $M$-graded Freyd category.
  \begin{align*}
    &I(\mathrm{id}_X) = \eta^T_X =  \mathrm{id}_X, \quad
      I(g \circ f) = (\eta^T \circ g)^\sharp \circ (\eta^T \circ f)  =  Ig \circ^{\CC_T} If,\\
    &\mathrm{id}_V \ast \mathrm{id}_X = \mathrm{st}^T_{V,X} \circ (\mathrm{id}_V \times \eta^T_X) = \eta^T_{V \times X}  = \mathrm{id}^{\CC_T}_{V \times X},\\
    &(g \circ f) \ast (i \circ^{\CC_T} j) = \mathrm{st}^T \circ (g \circ f \times i^\sharp \circ j) = (\mathrm{st}^T \circ (g \times i))^\sharp \circ \mathrm{st}^T \circ (f \times j) = (g \ast i) \circ^{\CC_T} (f \ast j), \\
    &f*\up X Y m n g = \mathrm{st}^T \circ (f \times {\Tup T m n} \circ g) = {\Tup T m n}  \circ \mathrm{st}^T(f \times g)  = \up {V*X} {W*Y} m n (f*g), \\
    &I(l_X)\circ^{\CC_T} (\id_1*f)
      = T(l_X) \circ  \mathrm{st}^T \circ (\mathrm{id}_1 \times f)
      = l_{TY} \circ (\mathrm{id}_1 \times f)
      = f \circ l_X
      = f \circ^{\CC_T} I(l_X)
    \\
    &I(a_{X',Y',Z'})\circ^{\CC_T} ((f\times g)*h)
      =T(a_{X',Y',Z'}) \circ  \mathrm{st}^T_{X' \times Y', TZ'}  \circ ((f\times g) \times h)\\
    &=\mathrm{st}^T_{X', Y' \times Z'} \circ (\mathrm{id}_{X'} \times \mathrm{st}^T_{Y', TZ'}) \circ a_{X',Y',TZ'} \circ ((f\times g) \times h)\\
    &=\mathrm{st}^T_{X', Y' \times Z'} \circ (\mathrm{id}_{X'} \times \mathrm{st}^T_{Y', TZ'})  \circ (f\times (g \times h)) \circ a_{X,Y,Z}
      =(f*(g*h)) \circ^{\CC_T} I(a_{X,Y,Z})
    \\
    &I(f \times g) = \eta^T \circ (f \times g) = \mathrm{st}^T \circ (f \times \eta^T \circ g) = f \ast Ig,
  \end{align*}
\end{appendixproof}

\begin{appendixproof}
  (Proof of Proposition \ref{pp:Klieslifreyd}) We next
  suppose that the underlying category $\CC$ has coproducts.  Then, the Kleisli grade category $\CC_T$ has homogeneous coproducts.
  From the coprojections
  $\iota_i \colon X_i \to \coprod_{i \in A} X_i$ in $\CC$, The tuple
  of morphisms $I(\iota_i) \colon X_i \to_1 \coprod_{i \in A} X_i$ in
  $\CC_T$ forms a countable homogeneous coproduct.

  \begin{align*}
    [f_i]_{i \in A} \circ^{\CC_T} I(\iota_i) &=
                                               [f_i]_{i \in A}^\sharp \circ (\eta^{T}_{ \coprod_{i \in I} X_i} \circ \iota_i) = [f_i]_{i \in A} \circ  \iota_i = f_i \\
    [I(\iota_i) ]_{i \in A} &= \eta^{T}_{ \coprod_{i \in I} X_i} \circ [\iota_i]_{i \in A} =  \eta^{T}_{ \coprod_{i \in A} X_i} = \mathrm{id}^{\CC_T}_{\coprod_{i \in A} X_i},\\
    h \circ^{\CC_T} [f_i]_{i \in A} & = h^\sharp \circ [f_i]_{i \in I} = [h^\sharp \circ f_i]_{i \in A} = [h \circ^{\CC_T} f_i]_{i \in I},\\
    [\uparrow_{m}^{n} (f_i)]_{i \in A} &= [\Tup T m n \circ f_i]_{i \in A}
                                         = \Tup T m n [f_i]_{i \in A} = \uparrow_m^n [f_i]_{i \in A}.
  \end{align*}
%
  From the construction of the homogeneous coproducts and the unit
  law of the tensorial strength
  ($\mathrm{st}^T_{X,Y} \circ (\mathrm{id}_X \times \eta^T_Y) =
  \eta^T_{X \times Y}$), $I$ preserves the distributivity.  Hence,
  $\mathrm{id}_V \ast I(\iota_i) = I(\mathrm{id}_V \times \iota_i)$,
  where
  $\mathrm{id}_V \times \iota_i \colon V \times X_i \to V \times
  \coprod_{i \in A} X_i$ are coprojections.  Hence
  $(\mathrm{id}_V \ast I(\iota_i) \colon V \times X_i \to_1 V \times
  \coprod_{i \in A} X_i)_{i \in A}$ also forms a countable
  homogeneous coproduct.  Since homogeneous coproduct are unique up
  to isomorphisms, we then conclude that $(\CC,1,\times,\CC_T,I,\ast)$
  has homogeneous countable coproducts.
\end{appendixproof}

\noindent
We often use the following `$\push{}$' operation to structure
interpretations of programs and GHL derivations. Let
$\delta_X\in\VV(X,X\times X)$ be the diagonal morphism.
Then $\push{}:\CC(X,Y)(m)\arrow \CC(X,X\times Y)(m)$ is defined as
$\push(f) = (X*f)\circ I\delta_X$.  When viewing $X$ as a set of
environments, $\push{(f)}$ may be seen as executing an effectful
procedure $f$ under an environment, then extending the environment with the return
value of $f$.  In a non-graded setting, the definition of $\push{}$ is
analogous.

\subsection{Semantics of \LangTitleTitle{} in Freyd
  Categories\conf{80}}\label{sec:semloop}

\newcommand{\framein}[1]{
  \noindent\fbox{\begin{minipage}{1.0\linewidth}#1\end{minipage}}
}

Towards the semantics of GHL, we first give a more
standard, non-graded categorical semantics of \Lang{}. We first
prepare the following data.
\begin{itemize}
\item[$\star$] A Freyd category $(\VV,1,\times,\CC,I,*)$ with
  countable coproducts.
\item[$\star$] A coproduct $\{\mtrue,\mfalse\in\VV(1,\bool)\}$ of $1$
  and $1$ in $\VV$.
  \item[$\star$] A coproduct $\{\mnat k\in\VV(1,\nat)\}_{k\in\natrep{}}$ of
    $\natrep$-many  $1$s in $\VV$.
\item[$\star$] An interpretation $\sem-$ of $\Sigma$ in $\VV$ such
  that
  \begin{align*}
    \sem\sbool&=\bool \qquad \sem\otrue=\mtrue\in\VV(1,\bool) \qquad \sem\ofalse=\mfalse\in\VV(1,\bool)\\
    \sem\snat&=\nat  \qquad \sem{\onat k}=\mnat k\in\VV(1,\nat).
  \end{align*}
\end{itemize}
For convenience, we let $\mem\teq \sem\mctx$ (Section
\ref{sec:prelim}), i.e., all relevant (mutable) program variables are in scope,
and write $\pi_v\in\VV(\mem,\sem{\mctx(v)})$ for
the projection morphism associated to a program variable $v\in\mctx$.

Pure expressions are interpreted as $\VV$-morphisms and
impure commands and procedures are interpreted as $\CC$-morphisms, of the form:
\begin{itemize}
\item[$\star$] (expressions) A morphism $\sem e \in \VV(\mem,
  \sem s)$ for all $e \in \Exp \Sigma \mctx s$; see Section \ref{sec:prelim}.
\item[$\star$] (commands) A morphism $\sem {c} \in\CC(\mem, 1)$ for each
  $c\in\CExp$.
\item[$\star$] (procedures) A morphism $\sem p\in\CC(\mem, \sem s)$ for each
  $s\in S$ and $p\in\PExp_s$.
\end{itemize}
For the interpretation of programs, we first define some auxiliary
morphisms. For all $v\in\mctx$, let
$\upd v\in\VV(\mem\times\sem{\mctx(v)},\mem)$ to be the unique
morphism (capturing memory updates) satisfying
$\pi_v\circ\upd v=\pi_2$ and $\pi_w\circ\upd v=\pi_w\circ\pi_1$ for
any $w\in\mctx$ such that $v\neq w$.  We define
$\sub v{e}\in\VV(\mem,\mem)$ by
$\sub v{e}\teq \upd v\circ\langle\id_\mem,\sem{e}\rangle$, which
updates the memory configuration at variable $v$ with the value of
$e$.

For the interpretation of conditional and loop commands, we need
coproducts over $\mem$.
Since $\VV$ is distributive, we can form a binary coproduct
$\mem \times \bool$ and a countable coproduct $\mem \times \nat$ with
injections respectively defined as ($\forall k \in \natrep{}$):
\begin{align*}
  & \ftrue \teq \langle\id_\mem,\mtrue\circ!_\mem\rangle
  \in\VV(\mem,\mem\times\bool)
  \quad
[k] \teq \langle\id_\mem,\mnat k\circ !_\mem\rangle
        \in\VV(\mem,\mem\times\nat)
  \\[-0.1em]
  & \ffalse \teq
  \langle\id_\mem,\mfalse\circ!_\mem\rangle
  \in\VV(\mem,\mem\times\bool)
\end{align*}
%
By Condition \ref{cond:icoprod} of
Definition \ref{def:freydcoprod}, these coproducts are mapped to
coproducts in $\CC$ with injections:
\begin{align*}
  & \{I(\ftrue),I(\ffalse)\in\CC(\mem,\mem\times\bool)\}, &
  & \{I(\fnat k)\in\CC(\mem,\mem\times\nat)~|~k\in\natrep{}\}.
\end{align*}
The cotuplings of these coproducts (written $[f, g]$ and $[ f^{(k)} ]_{k
  \in \natrep{}}$ respectively) are used next to interpret
conditionals and loops.

\noindent
We interpret a program $P$ of \Lang{} as a morphism
$\sem{P}\in\CC(\mem,\mem)$:
\begin{align*}
\setlength{\arraycolsep}{0.1em}
\begin{array}{rlrl}
  \sem{P;P'} & = \sem{P'}\circ\sem P &
  \sem{\synSkip{}} & =\id_\mem \\
  \sem{\synProc v p} & = I(\upd v)\circ\push{\sem p} &
  \sem{\synComm c} & = I(\pi_1)\circ\push{\sem c} \\
  \sem{\synExp v {\exprc{}}} & = I(\sub v{\exprc{}}) \\ 
  \sem{\synIf{e_b} P {P'}} & = [\,\sem P,\sem{P'}\,]\circ
  \push{(I\sem{e_b})} \\
  \sem{\synLoop{e_n}{P}} & = [\,\sem
  P^{(k)}\,]_{k\in\natrep{}}\circ\push{(I\sem{e_n})}
\end{array}
\end{align*}
%
%
Thus, the semantics of $\synLoop{e_n}{P}$ is such that, if the expression
$e_n$ evaluates to some natural number $\onat k$ then
$\synLoop{e_n}{P}$ is equivalent to the $k$-times sequential composition of $P$.

\section{Modelling Graded Hoare Logic}
\label{sec:model}

\newcommand{\actionTensor}{\circledast}

We now define the categorical model of GHL, building on the
non-graded Freyd semantics of Section~\ref{sec:semloop}.
Section~\ref{sec:interp-assertion-logic} first models the base
assertion logic, for which we use fibrations, giving an overview
of the necessary mathematical machinery for completeness.
Section~\ref{subsection:interpretation:GHL} then defines
the semantics of GHL and Section~\ref{sec:instances} instantiates
it for the examples discussed previously in Section~\ref{sec:main}.

\subsection{Interpretation of the Assertion Logic using Fibrations \conf{80}}
\label{sec:interp-assertion-logic}

Our assertion logic (Section~\ref{sec:main}) has logical
connectives of finite conjunctions, countable disjunctions,
existential quantification and an equality predicate. A suitable
categorical model for this fragment of first-order logic is offered by
a {\em coherent fibration} \cite[Def.~4.2.1]{jacobscltt},
extended with countable joins in each fibre.
We recap various key definitions and
terminology due to Jacobs' textbook~\cite{jacobscltt}.

In the following, let $\PP$ and
$\VV$ be categories and $p:\PP\arrow\VV$ a functor.

We can regard
functor $p$ as attaching {\em predicates} to each object in
$\VV$. When $p\psi=X$, we regard $\psi\in\PP$ as
a predicate over $X\in\VV$. When $f\in\PP(\psi,\phi)$ is a
morphism, we regard this as saying that $p f$ maps elements
satisfying $\psi$ to those satisfying $\phi$ in $\VV$.
%
Parallel to this view of functors assigning predicates
is the notion that entities in $\PP$ are `above' those in $\VV$
when they are mapped to by $p$.

\begin{defn}[`Aboveness']
  An object $\psi\in\PP$
  is said to be {\em above} an object $X\in\VV$
  if $p\psi=X$. Similarly, a morphism\footnote{The dot
    notation here introduces a new name and should not
    be understood as applying some mathematical operator on $f$.} $\dot f \in
  \PP(\psi, \phi)$ is said
  to be {\em above} a morphism $f$ in $\VV$ if $p\dot f=f \in \VV(p \psi, p \phi)$.  A morphism
  in $\PP$ is {\em vertical} if it is above an identity morphism.
  Given $\psi,\phi\in\PP$ and
  $f\in\VV(p\psi,p\phi)$, then we denote
  the \emph{set of all morphisms in $\PP$ above
    $f$} as $\PP_f(\psi,\phi) =
  \{\dot f\in\PP(\psi,\phi)~|~p\dot f=f\}$.
\end{defn}

\begin{defn}[Fibre category]
  A {\em fibre category} over $X\in\VV$ is a subcategory of $\PP$
  consisting of objects above $X$ and morphisms above $\id_X$.
  This subcategory is denoted by $\PP_X$, and thus
  the homsets of $\PP_X$ are
  $\PP_{X}(\psi, \phi) = \PP_{\id_X}(\psi, \phi)$.
\end{defn}

\noindent
We are ready to recall the central concept in fibrations:
{\em cartesian morphisms}.
\begin{defn}[Cartesian morphism]
  A morphism $\dot f\in\PP(\psi,\phi)$ is {\em
    cartesian} if for any $\alpha\in\PP$ and
  $g\in\VV(p\alpha,p\psi)$,
  the post-composition of $\dot f$ in $\PP$, regarded as a function of type
  $
    \dot f\circ -:\PP_g(\alpha,\psi)\arrow \PP_{g\circ p\dot f}(\alpha,\phi),
  $
  is a bijection. This amounts to the following {\em universal property}
  of cartesian morphism: for any $\dot h\in\PP(\alpha,\phi)$ above
  $g\circ pf$, there exists a unique morphism $\dot g\in\PP(\alpha,\psi)$
  above $g$ such that $\dot h=\dot f\circ\dot g$. Intuitively, $\dot f$
  represents the situation where $\psi$ is a {\em pullback} or {\em
    inverse image} of $\phi$ along $p\dot f$, and the universal property
  corresponds to that of pullback. 
  \end{defn}

\begin{defn}[Fibration]
  Finally, a functor $p:\PP\arrow\VV$ is a {\em fibration}
  if for any $\psi\in\PP$, $X \in \VV$, and $f\in\VV(X,p\psi)$, there exists an object
  $\phi\in\PP$ and a cartesian morphism $\dot f\in\PP(\phi,\psi)$ above
  $f$, called the {\em cartesian lifting} of $f$ with $\psi$. 
  We say that a fibration $p:\PP\arrow\VV$ is {\em posetal} if each
  $\PP_X$ is a poset, corresponding to the implicational order
  between predicates. When $\psi\le\phi$ holds in $\PP_X$, we denote the
  corresponding vertical morphism in $\PP$ as $\psi\lev \phi$.

  Posetal fibrations are always faithful. The cartesian lifting of
  $f\in\VV(X,p\psi)$ with $\psi$ uniquely exists. We thus write it by
  $\ol f\psi$, and its domain by $f^*\psi$.  It can be easily shown that
  for any morphism $f\in\VV(X,Y)$ in $\VV$, the assignment
  $\psi\in\PP_Y\mapsto f^*\psi\in\PP_X$ extends to a monotone function
  $f^*:\PP_Y\arrow\PP_X$. We call it the {\em reindexing function}
  (along $f$). Furthermore, the assignment $f\mapsto f^*$ satisfies the
  (contravariant) functoriality: $\id_X^*=\id_{\PP_X}$ and
  $(g\circ f)^*=f^*\circ g^*$. 
  A fibration is a {\em bifibration} if each
  reindexing function $f^*:\PP_Y\arrow\PP_X$ for $f\in\VV(X,Y)$ has a
  left adjoint, denoted by $f_*:\PP_X\arrow\PP_Y$.  $f_*\psi$ is always
  associated with a morphism $\ul f\psi:f_*\psi\arrow\psi$ above $f$,
  and this is called the {\em opcartesian lifting} of $f$ with $\psi$.
  For the universal property of the opcartesian lifting, see Jacobs~\cite[Def. 9.1.1]{jacobscltt}.
\end{defn}


\paragraph{Fibrations for our Assertion Logic}
It is widely known that {\em coherent fibrations} are suitable for
interpreting the $\wedge,\vee,\exists,{=}$-fragment of first-order
logic (see \cite[Chapter 4, Def. 4.2.1]{jacobscltt}).  Based on this
fact, we introduce a class of fibrations that are suitable for our
assertion logic---due to the countable joins of the assertion logic
we modify the definition of coherent fibration accordingly.
\begin{definition}
  A {\em fibration for assertion logic} 
  over $\VV$ is a
  posetal fibration $p:\PP\arrow\VV$ for cartesian $\VV$ with
  distributive countable coproducts, such that:
  \begin{enumerate}
  \item Each fibre poset $\PP_X$ is a distributive lattice with finite
    meets $\top_X,\wedge$ and countable joins $\bot_X,\vee$.
  \item Each reindexing function $f^*$ preserves finite meets and
    countable joins.
  \item The reindexing function $c_{X,Y}^*$ along the contraction
    $c_{X,Y}\teq\langle\pi_1,\pi_2, \pi_2\rangle \in \VV(X\times Y,
    X\times Y\times Y)$ has a left adjoint
    $\Eq_{X,Y}\dashv c_{X,Y}^*$.
    This satisfies {\em Beck-Chevalley
      condition} and {\em Frobenius property}; we refer
    to \cite[Definition 3.4.1]{jacobscltt}.

  \item The reindexing function $w_{X,Y}^*$ along the weakening
    $w_{X,Y}\teq\pi_1\in\VV(X\times Y,X)$ has a left adjoint
    $\exists_{X,Y}\dashv w_{X,Y}^*$. This satisfies {\em
      Beck-Chevalley condition} and {\em Frobenius property};
    we refer
    \cite[Definition 1.9.1, 1.9.12]{jacobscltt}.
  \end{enumerate}
\end{definition}
\noindent
This is almost the same as the definition of coherent fibrations
\cite[Definition 4.2.1]{jacobscltt}; the difference is that 1) the
base category $\VV$ has countable coproducts 2) we require each fibre
to be a poset; this makes object equalities hold on-the-nose, and 3) we
require each fibre to have countable joins. They will be combined with
countable coproducts of $\VV$ to equip $\PP$ with a countable
coproduct \cite{jacobscltt}. 
\begin{example}\label{ex:pred}
  A typical example of a fibration for assertion logic is the
  subobject fibration $\predfib{\Set}:\Pred\arrow\Set$; the category $\Pred$ has
  objects pairs $(X,\psi)$ of sets such that $\psi\subseteq X$, and
  morphisms of type $(X,\psi)\arrow(Y,\phi)$ as functions
  $f:X\arrow Y$ such that $f(\psi)\subseteq\phi$.
  The functor $\predfib{}$ sends $(X,\psi)$ to $X$ and $f$ to itself.
  More examples can be found in the work of Jacobs~\cite[Section 4]{jacobscltt}.
\end{example}

For a parallel pair of morphisms $f,g\in\VV(X,Y)$, we define the
equality predicate $\Eq(f,g)$ above $X$ to be
$\langle\id_X,f,g\rangle^*\Eq_{X,Y}(\top_{X\times Y})$ \cite[Notation
3.4.2]{jacobscltt}.  Intuitively, $\Eq(f,g)$ corresponds to the
predicate $\{x\in X~|~f(x)=g(x)\}$.  In this paper, we will use some
facts about the equality predicate shown by Jacobs~\cite[Proposition
3.4.6, Lemma 3.4.5, Notation 3.4.2, Example 4.3.7]{jacobscltt}.
\marginpar{\dnote{Hmm mysterious. What are these? We should at least
    point to where we use them}}

\begin{toappendix}
We formally state various known facts in coherent posetal fibrations:
\begin{proposition}\label{pp:coh}
  For any morphism $f,g\in\VV(X,Y)$, the following holds.
  \begin{enumerate}
  \item \cite[Proposition 3.4.6]{jacobscltt} \label{rem:rem346}
    $\Eq (f, g) = \Eq (\langle \id, f \rangle, \langle\id,g\rangle)$.
  \item \cite[Lemma 3.4.5]{jacobscltt} \label{rem:rem345}
    For any $\psi\in\PP_Y$,
    $f^{\ast} \psi \wedge \Eq (f, g) \leqslant g^{\ast} \psi$.
  \item \cite[Notation 3.4.2]{jacobscltt}
    For any $h\in\VV(Z,X)$,
    $h^*\Eq(f,g)=\Eq(f\circ h,g\circ h)$.
  \item \cite[Example 4.3.7]{jacobscltt} For $f\in\VV(X,Y)$, the
    function $f_*:\PP_X\arrow\PP_Y$ defined below is a left adjoint to
    the reindexing function $f^*:\PP_Y\arrow\PP_X$ (hence $p$ is a
    bifibration).
    \begin{equation}
      \label{eq:img}
      f_{\ast} \psi \triangleq \exists_{Y,X}((\pi_{Y, X}')^{\ast} \psi \wedge
      \Eq (f \circ \pi_{Y, X}', \pi_{Y, X})) .
    \end{equation}
  \item \cite{jacobscltt} \label{lem:coprod} The category $\PP$ has
    finite products and distributive countable coproducts that are
    strictly preserved by $p$.
  \end{enumerate}
\end{proposition}
\begin{proof}
  We briefly prove point (5). Let $\{\psi_i\}_{i\in\natrep{}}\in\PP$.
  For a countable coproduct
  $\{\iota_i:p \psi_i\arrow X\}_{i\in\natrep{}}$ in $\VV$, the following
  family is a countable coproduct in $\PP$ above coproduct
  $\{\iota_i\}_{i\in\natrep{}}$:
  $
  \{
  \xymatrix{
    \psi_i \rrrh{\ul{\iota_i}(\psi_i)} & &
    (\iota_i)_*\psi_i \rrrh{\lev} & &
    \bigvee_{i\in\natrep{}}(\iota_i)_* \psi_i
  }\}_{i\in\natrep{}}
  $.
  \qed
\end{proof}
\end{toappendix}

\subsubsection{The Semantics of Assertion Logic\conf{80}}

We move to the semantics of our assertion logic in a fibration
$p:\PP\arrow\VV$ for assertion logic. The basic idea is to interpret a
formula $\psi\in\Fml{\lsig}\Gamma$ as an object in $\PP_{\sem\Gamma}$,
and an entailment $\Gamma~|~\psi\vdash\phi$ as the order
relation $\sem\psi\le\sem\phi$ in $\PP_{\sem\Gamma}$. The semantics is
given by the following interpretation of the data specifying the
assertion logic (given in Section~\ref{sec:blogic}):
\begin{itemize}
\item [$\star$] A fibration $p:\PP\arrow\VV$ for assertion logic.
\item [$\star$] An interpretation $\sem-$ of $\lsig$ in $\PP$ that
  coincides with the one $\sem-$ of $\Sigma$ in $\VV$.
\item [$\star$] An object $\sem P\in\PP_{\sem{\mathit{par}(P) }}$ for each
  atomic proposition $P\in P_l$ (recall $\mathit{par}$ assigns
  input sorts to atomic propositions in $P_l$, parameterising
  the logic).
\item [$\star$] We require that for any $\Gamma\in\Ctx\lsig$ and
  $(\psi,\phi)\in\Axiom(\Gamma)$, $\sem\psi\le\sem\phi$ holds in
  $\PP_{\sem\Gamma}$. This expresses an implicational axiom in the coherent logic.
\end{itemize}
The interpretation $\sem\psi$ of
$\psi\in\Fml\lsig\Gamma$ is inductively defined as a $\PP_{\sem\Gamma}$-object:
\begin{align*}
  & \textstyle
    \sem{P(t_1,\cdots,t_n)}
    =\langle\sem{t_1},\cdots,\sem{t_n}\rangle^*\sem P
    \qquad
    \sem{t=u}
    =\Eq(\sem t,\sem u)
  \\
  & \textstyle
    \sem{\bigwedge\psi_i}
    =\bigwedge\sem{\psi_i}
    \qquad
    \sem{\bigvee\psi_i}
    =\bigvee\sem{\psi_i}
    \qquad
    \sem{\ex{x:s}\psi}
    =\exists_{\sem\Gamma,\sem s}\sem\psi
\end{align*}

\subsection{Interpretation of Graded Hoare Logic\conf{80}}
\label{subsection:interpretation:GHL}
\begin{toappendix}
	\subsection*{Interpretation of Graded Hoare Logic}
\end{toappendix}

We finally introduce the semantics of Graded Hoare logic.  This
semantics interprets derivations of GHL judgements
$\tripleV{m}{\psi}{P}{\phi}$ as $m$-graded morphisms in a graded
category. Moreover, it is built {\em above} the interpretation
$\sem P\in\CC(\mem,\mem)$ of the program $P$ in the non-graded
semantics introduced in Section \ref{sec:semloop}.
The underlying structure is given as a combination of a
fibration for the assertion logic and a graded category over $\CC$, as depicted
in \eqref{eq:square} (Section~\ref{sec:overview}, p.~\pageref{eq:square}).
\begin{definition}\label{def:logstr}
  A {\em \ghlstr{}} over a Freyd category
  $(\VV,1,\times,\CC,I,*)$ with countable coproducts and a fibration
  $p:\PP\arrow\VV$ for assertion logic comprises:
  \begin{enumerate}
  \item An $M$-graded Freyd category
    $(\PP,\dot 1,\dtimes,\EE,\dot I,\actionTensor)$ with homogeneous
    countable coproducts.

  \item A function
    $q_{\psi,\phi,m}:\EE(\psi,\phi)(m)\arrow\CC(p\psi,p\phi)$ (subscripts may
    be omitted), which maps to the base denotational model,
  erasing assertions and grades.
  \end{enumerate}
  The above data satisfy the following properties:
  \begin{enumerate}
  \item That $q$ behaves `functorialy' preserving structure from $\EE$ to $\VV$:
  \begin{align*}
    & q(\id_\phi)=\id_{p\phi},\quad q(g\circ f)=qg\circ qf,\quad 
      q_{\psi,\phi,n}(\up \psi \phi m n f) = q_{\psi,\phi,m} f\\[-0.4em]
    & q(\dot I f)=I(p f),\quad q(f\actionTensor{} g)=p f * q g
  \end{align*}
  \item For any homogeneous countable coproduct
    $\{\iota_i\in\EE(\psi_i,\phi)(1)\}_{i\in\Lambda}$,
    $\{q\iota_i\in\CC(p\psi_i,p\phi)\}_{i\in\Lambda}$ is a countable
    coproduct.
  \item (Ex falso quodlibet) $q_{\bot_X,\phi,m}:\EE(\bot_X,\phi)(m)\arrow\CC(X,p\phi)$ is a
    bijection.
  \end{enumerate}
\end{definition}
\noindent
The last statement asserts that if the precondition is the least
element $\bot_X$ in the fibre over $X\in\VV$, which represents the
false assertion, we trivially conclude any postcondition $\phi$
and grading $m$ for any morphisms of type $X\arrow p\phi$ in $\CC$.

\newcommand{\esem}[1]{\langle #1\rangle}

The semantics of GHL then requires a
graded Freyd category with countable coproducts, and morphisms in the
graded category guaranteeing a sound model of the effectful primitives
(commands/procedures), captured by the data:
\begin{itemize}[leftmargin=1em]
\item [$\star$] A \ghlstr{}
  $(\PP,\dot 1,\dtimes,\EE,\dot I,\actionTensor,q)$ over
  the Freyd category $(\VV,1,\times,\CC,I,*)$ with countable
  coproducts and the fibration $p:\PP\arrow\VV$ for assertion logic.

\item [$\star$] For each $c\in \availcom(\psi,m)$ a morphism
  $\esem c\in\EE(\sem\psi,\dot 1)(m)$ such that
  $q\esem c=\sem c$.

\item [$\star$] For each $p\!\in\!\availproc s(\psi,m,\phi)$ a morphism
  ${\esem p}\!\in\!\EE(\sem\psi,\sem\phi)(m)$ such that $q\esem p=\sem p$.
\end{itemize}
where $\sem{c}$, $\sem{p}$ and later $\sem{e}$ are from the underlying non-graded model
(Sec.~\ref{sec:semloop}).

We interpret a derivation of GHL judgement
$\tripleV{m}{\phi}{P}{\psi}$ as a morphism
\begin{displaymath}
  \sem{\tripleV{m}{\phi}{P}{\psi}}\in\EE(\sem\phi,\sem\psi)(m)
  ~~
  \text{such that}
  ~~
  q_{\sem\phi,\sem\psi,m}\sem{\tripleV{m}{\phi}{P}{\psi}}=\sem P.
\end{displaymath}
The constraint on the right is guaranteed by the soundness of the
interpretation (Theorem \ref{th:sound}).  From the
functor-as-refinement viewpoint \cite{DBLP:conf/popl/MelliesZ15}, the
interpretation $\sem{\tripleV{m}{\phi}{P}{\psi}}$ witnesses that
$\sem P$ respects refinements $\phi$ and $\psi$ of $\mem$, and
additionally it witnesses the grade of $\sem P$ being $m$.
We first cover the simpler cases of the interpretation of GHL
derivations:
\begin{align*}
  \sem{\tripleV 1\psi\synSkip\psi}
  &= \id_{\sem\psi} \\
  \sem{\tripleV{m_1 \cdot m_2}\psi{\synSeq{P_1}{P_2}}\theta}
  &= \sem{\tripleV{m_2}{\psi_1}{P_2}{\theta}}\circ\sem{\tripleV{m_1}{\psi}{P_1}{\psi_1}} \\
  \sem{\tripleV{1}{\psi[\exprc{}/v]}{\synExp v{\exprc{}}}{\psi}}
  &= \dot I(\ol{\sub v{\exprc{}}}\sem\psi) \\
  \sem{\tripleV{m}{\psi}{\synComm c}{\psi}}
  &=\dot I(\pi_1)\circ\push{\esem c}\\
  \sem{\tripleV m \psi {\synProc v p} {(\ex v\psi)\wedge\phi}}
  &=\dot I(\ul{\upd v}(\sem\psi\dtimes\sem\phi))\circ\push{\esem p} \\
  \lefteqn{\hspace{-14em} \sem{\tripleV {m'}{\psi'}{P}{\phi'}}
  =\dot I(\sem\phi\lev\sem{\phi'})\circ \up {\sem\psi}{\sem\phi} m {m'}\sem{\tripleV {m}{\psi}{P}{\phi}}\circ\dot I(\sem{\psi'}\lev\sem\psi)
  }
\end{align*}
The morphisms with upper and lower lines are cartesian
liftings and op-cartesian liftings in the fibration $p:\PP\arrow\VV$
of the assertion logic. The codomain of the interpretation of the
procedure call $\synProc v p$ is equal to
$\sem{(\ex v\psi)\wedge\phi}$.

The above interpretations largely follow the form of the underlying model of
Section~\ref{sec:semloop}, with the additional information and
underlying categorical machinery for grades and assertions here; we
now map to $\EE$. The interpretation of
conditional and loop commands requires some more reasoning.

\renewcommand{\Im}{\mathsf{Im}}
\paragraph{Conditionals}
Let $p_1,p_2$ be the interpretations of each branch of the
conditional command:
\begin{align*}
  p_1&=\sem{\tripleV m {\psi\wedge e_b=\otrue} {P_1} \phi}
       \in\EE(\sem{\psi\wedge e_b=\otrue},\sem\phi)(m)
  \\
  p_2&=\sem{\tripleV m {\psi\wedge e_b=\ofalse} {P_2} \phi}
       \in\EE(\sem{\psi\wedge e_b=\ofalse},\sem\phi)(m)
\end{align*}
We consider the cocartesian lifting
$\ul{\langle\id_\mem,\sem{e_b}\rangle}\sem\psi:\sem\psi\arrow
\langle\id_\mem,\sem{e_b}\rangle_*\sem\psi$. We name its codomain
$\Im$. Next, cartesian morphisms
$\ol{\ftrue}(\Im):\ftrue^*\Im\arrow \Im$ and
$\ol{\ffalse}(\Im):\ffalse^*\Im\arrow \Im$ in $\PP$ are above the
coproduct $(\mem\times\bool,\ftrue,\ffalse)$ in $\VV$. Then the
interpretations of the preconditions of $P_1,P_2$ are inverse images
of $\Im$ along $\ftrue,\ffalse:\mem\arrow\mem\times\bool$:
\begin{lemma}
  $\sem{\psi\wedge e_b=\otrue}=\ftrue^*\Im$ and
  $\sem{\psi\wedge e_b=\ofalse}=\ffalse^*\Im$.
\end{lemma}
\begin{appendixproof}
  In this proof $\sem\psi,\sem{e_b}$ are abbreviated to $\psi,e_b$
  respectively. First, we obtain one direction:
  \begin{eqnarray*}
    & & \psi \wedge \Eq (e_b, \mtrue\circ {!})\\
    \text{Proposition \ref{pp:coh}-\ref{rem:rem346}}
    & = & \psi \wedge \Eq (\langle\id,e_b\rangle, \ftrue)\\
    & \leqslant & \langle\id,e_b\rangle^{\ast} \Im \wedge \Eq (\langle\id,e_b\rangle, \ftrue)\\
    \text{Proposition \ref{pp:coh}-\ref{rem:rem345}}
    & \leqslant & \ftrue^{\ast} \Im .
  \end{eqnarray*}
  We show the other direction
  $\ftrue^{\ast} \Im \leqslant \psi \wedge \Eq (e_b,
  \mtrue\circ {!})$.
  \begin{eqnarray*}
    & & \ftrue^{\ast} \Im\\
    \text{Pushforward \eqref{eq:img}}
    & = & \ftrue^{\ast} (\pi')_{\ast} (\pi^{\ast} \psi \wedge \Eq (\langle\id,e_b\rangle \circ \pi, \pi'))\\
    \text{Beck-Chevalley}& = & (\pi')_{\ast} (\mem \times \ftrue)^{\ast} (\pi^{\ast} \psi\wedge \Eq (\langle\id,e_b\rangle \circ \pi, \pi'))\\
    & = & (\pi')_{\ast} ((\mem \times \ftrue)^{\ast}\pi^{\ast} \psi\wedge (\mem \times \ftrue)^{\ast}\Eq (\langle\id,e_b\rangle \circ \pi, \pi'))\\
    \text{PBEQ}& = & (\pi')_{\ast} (\pi^{\ast} \psi \wedge \Eq (\langle\id,e_b\rangle \circ \pi, \ftrue \circ \pi'))\\
    & = & (\pi')_{\ast} (\pi^{\ast} \psi \wedge \Eq (\langle \pi, e_b \circ \pi \rangle, \langle \pi', \mtrue\circ {!} \circ \pi \rangle))\\
    \text{Proposition \ref{pp:coh}-\ref{rem:rem346}}
    & = & (\pi')_{\ast} (\pi^{\ast} \psi \wedge \Eq (e_b \circ \pi, \mtrue\circ {!} \circ \pi') \wedge \Eq (\pi, \pi'))\\
    & = & (\pi')_{\ast} (\pi^{\ast} \psi \wedge \Eq (e_b \circ \pi, \mtrue\circ {!} \circ \pi) \wedge \Eq (\pi, \pi'))\\
    \text{PBEQ}& = & (\pi')_{\ast} (\pi^{\ast} \psi \wedge \pi^*\Eq (e_b, \mtrue\circ {!}) \wedge \Eq (\pi, \pi'))\\
    \text{Proposition \ref{pp:coh}-\ref{rem:rem345}}
    & = & (\pi')_{\ast} ((\pi')^{\ast} \psi \wedge (\pi')^{\ast} \Eq (e_b, \mtrue\circ {!}))\\
    & \leqslant & \psi \wedge \Eq (e_b, \mtrue\circ {!}) .
  \end{eqnarray*}
\end{appendixproof}
\noindent
The side condition of the conditional rule ensures that
$(\Im,\ol{\ftrue}(\Im),\ol{\ffalse}(\Im))$ is a coproduct in $\PP$:
\begin{lemma}
  $\mctx~|~\psi\vdash e_b=\otrue\vee e_b=\ofalse$ implies
  $\Im=\ftrue_*\ftrue^*\Im\vee\ffalse_*\ffalse^*\Im$.
\end{lemma}
\begin{appendixproof}
  In this proof we abbreviate $\sem\psi,\sem{e_b}$ to $\psi,e_b$ respectively.
  The direction $\ftrue_*\ftrue^*\Im\vee\ffalse_*\ffalse^*\Im\le\Im$ is
  immediate. We thus show the converse.  Assume
  $\psi \leqslant \Eq (e_b, \mtrue \circ !) \vee \Eq (e_b, \mfalse
  \circ !) . $ This is equivalent to
  \[ \psi = (\psi \wedge \Eq (e_b, \mtrue \circ !)) \vee (\psi \wedge
    \Eq (e_b, \mfalse \circ !)) . \] By applying
  $\langle \id, e \rangle_{\ast}$ both sides, which preserves joins,
  we obtain
  \[ \Im = \langle \id, e_b \rangle_{\ast} \psi = \langle
    \id, e_b \rangle_{\ast} (\psi \wedge \Eq (e_b, \mtrue
    \circ !)) \vee \langle \id, e_b \rangle_{\ast} (\psi \wedge
    \Eq (e_b, \mfalse \circ !)) ; \]
  From this, we aim to show
  \[ \langle \id, e_b \rangle_{\ast} (\psi \wedge \Eq (e_b,
    \mtrue \circ !)) \leqslant \ftrue_{\ast} \ftrue^{\ast}
    \Im ; \]
  the case for $\mfalse$ and $\ffalse$ is similar. Using the adjoint
  mate, the goal is equivalent to
  \[ \psi \wedge \Eq (e_b, \mtrue \circ !) \leqslant \langle
    \id, e_b \rangle^{\ast} \ftrue_{\ast} \ftrue^{\ast} \langle
    \id, e_b \rangle_{\ast} \psi . \]

  Then we proceed as follows.
  \begin{eqnarray*}
    & & \psi \wedge \Eq (e_b, \mtrue \circ !)\\
    \text{Proposition \ref{pp:coh}-\ref{rem:rem346}}
    & = & \psi \wedge \Eq (\langle \id, e_b \rangle, \ftrue) \wedge \Eq (\ftrue, \langle \id, e_b \rangle)\\
    & \leqslant & \langle \id, e_b \rangle^{\ast} \langle \id, e_b
                  \rangle_{\ast} \psi \wedge \Eq (\langle \id, e_b \rangle,
                  \ftrue) \wedge \Eq (\ftrue, \langle \id, e_b \rangle)\\
    \text{Proposition \ref{pp:coh}-\ref{rem:rem345}}
    & \leqslant & \ftrue^{\ast} \langle \id, e_b \rangle_{\ast}
                                                                  \psi \wedge \Eq (\ftrue, \langle \id, e_b \rangle)\\
    & \leqslant & \ftrue^{\ast} \ftrue_{\ast} \ftrue^{\ast} \langle
                  \id, e_b \rangle_{\ast} \psi \wedge \Eq (\ftrue, \langle
                  \id, e_b \rangle)\\
    \text{Proposition \ref{pp:coh}-\ref{rem:rem345}}
    & \leqslant & \langle \id, e_b \rangle^{\ast} \ftrue_{\ast}
                                                                  \ftrue^{\ast} \langle \id, e_b \rangle_{\ast} \psi .
  \end{eqnarray*}
\end{appendixproof}
\noindent
Therefore the image of the coproduct
$(\Im,\ol{\ftrue}(\Im),\ol{\ffalse}(\Im))$ by $\dot I$ yields a
homogeneous coproduct in $\EE$. We take the cotupling
$[p_1,p_2]\in\EE(\Im,\sem\phi)(m)$ with respect to this homogeneous
coproduct. We finally define the interpretation of the conditional
rule to be the following composite:
\begin{displaymath}
  \sem{\tripleV{m}{\psi}{\synIf{e_b}{P_1}{P_2}}{\phi}}
  =[p_1,p_2]\circ\dot I(\ul{\langle\id_\mem,\sem{e_b}\rangle}\sem\psi)
  \in \EE(\sem\psi,\sem\phi)(m).
\end{displaymath}

\paragraph{Loops}
Fix $N\in\natrep{}$, and suppose that $\tripleV{m}{\psi_{i+1}}P{\psi_i}$
is derivable in the graded Hoare logic for each $0\le i < N$. Let
$p_i\in\EE(\sem{\psi_{i+1}},\sem{\psi_i})(m)$ be the interpretation
$\sem{\tripleV{m}{\psi_{i+1}}{P_i}{\psi_i}}$. We then define a
countable family of morphisms (we use here ex falso quodlibet):
\begin{displaymath}
  b_i=
  \begin{choice}
    q^{-1}_{\bot_\mem,\sem{\psi_0},m^N}(\sem P^{(i)}) &\in\EE(\bot_\mem,\sem{\psi_0})(m^N)\quad (i\neq N) \\
    p_0\circ\cdots\circ p_N &\in\EE(\sem{\psi_N},\sem{\psi_0})(m^N)\quad (i=N)
  \end{choice}
\end{displaymath}
Let $\theta_i\teq\cod(b_i)$.  Then
$\coprod_{i\in\natrep{}}\theta_i=\bigvee_{i\in\natrep{}} \fnat
i_*\theta_i=\fnat N_*\sem{\psi_N}$ because $\fnat i_*\theta_i$ is
either $\bot_{\mem\times\nat}$ or $\fnat N_*\sem{\psi_N}$.  We then
send the coproduct $\theta_i\arrow\coprod_{i\in\natrep{}}\theta_i$ by
$\dot I$ and obtain a homogeneous coproduct in $\EE$. By taking the
cotupling of all $b_i$ with this homogeneous coproduct, we obtain a
morphism
$[b_i]_{i\in\natrep{}}\in\EE(\fnat N_*\sem{\psi_N},\sem{\psi_0})(m^N)$.
\begin{lemma}
  $\mctx~|~\psi_N\vdash e_n=\onat N$ implies
  $\langle\id_\mem,\sem{e_N}\rangle_*\sem{\psi_N}=\fnat
  N_*\sem{\psi_N}$.
\end{lemma}
\begin{appendixproof}
  Assume $\sem{\psi_N}\le\Eq(\sem{e_n},\mnat N)$. Then we proceed
  as follows:
  \begin{align*}
    \sem{\psi_N}
    & = \sem{\psi_N}\wedge \Eq(\sem{e_n},\mnat N) \\
    \text{Proposition \ref{pp:coh}-\ref{rem:rem346}}
    & = \sem{\psi_N}\wedge\Eq(\langle\id_\mem,\sem{e_n}\rangle,\fnat N) \\
    & \le \langle\id_\mem,\sem{e_n}\rangle^*\langle\id_\mem,\sem{e_n}\rangle_*\sem{\psi_N}\wedge\Eq(\langle\id_\mem,\sem{e_n}\rangle,\fnat N) \\
    \text{Proposition \ref{pp:coh}-\ref{rem:rem345}}
    & \le \fnat N^*\langle\id_\mem,\sem{e_n}\rangle_*\sem{\psi_N}.
  \end{align*}
  We thus conclude $\fnat N_*\sem{\psi_N}\le\langle\id_\mem,\sem{e_n}\rangle_*\sem{\psi_N}$. We obtain the other direction similarly.
\end{appendixproof}
\noindent
We then define
$ \sem{\tripleV{m^N}{\psi_N}{\synLoop{e_n}{P}}{\psi_0}}=
[b_i]_{i\in\natrep{}}\circ\dot
I(\ul{\langle\id_\mem,\sem{e_n}\rangle}\sem{\psi_N}).  $
\begin{theorem}[Soundness of GHL]\label{th:sound}
  For any derivation of a GHL judgement
  $\tripleV{m}{\phi}{P}{\psi}$, we have
  $q_{\sem\phi,\sem\psi,m}\sem{\tripleV{m}{\phi}{P}{\psi}}=\sem P$.
\end{theorem}
\begin{appendixproof}
  By induction over the interpretation.
\end{appendixproof}

\subsection{Instances of Graded Hoare Logic}
\label{sec:instances}
\begin{toappendix}
	\subsection{Further instances of Graded Hoare Logic}
\label{app:instances}
\end{toappendix}



We first present a construction of \ghlstrs{}
from graded monad liftings,
which are a graded version of the concept of {\em monad lifting}
\cite{filinskiphd,DBLP:journals/mscs/Goubault-LarrecqLN08,DBLP:journals/lmcs/KatsumataSU18}.

\begin{definition}[Graded Liftings of Monads]
  Consider two cartesian categories $\EE$ and $\CC$ and a functor
  $q \colon \EE \to \CC$ strictly preserving finite products.  We say that a
  strong $M$-graded monad $(\dot T,\dot \eta,\dot
  \mu_{m,m'},
  \dot{\mathrm{st}}_{m})$ on $\EE$ is an {\em $M$-graded lifting} of a
  strong monad $(T,\eta^T,\mu,\mathrm{st})$ on $\CC$ along $q$ if
  $q\circ\dot T m = T\circ q$,
  $q(\dot \eta_\psi) = \eta_{q\psi}$,
  $q(\dot \mu_{m,m',\psi}) = \mu_{q\psi}$,
  $q(\dot T(m_1\leq m_2)_\psi) = \mathrm{id}$,
  $q(\dot{\mathrm{st}}_{\psi,\phi,m}) =\mathrm{st}_{q\psi,q\phi}$.
    %

\end{definition}

\begin{theorem}
  \label{theorem:logical_structure:graded_lifting}
  Let $\VV$ be cartesian category with distributive countable
  coproducts, and let $p \colon \PP \to \VV$ be a fibration for
  assertion logic.  Let $T$ be a strong monad on $\VV$ and $\dot T$ be an
  $M$-graded lifting of $T$ along $p$. Then the $M$-graded Freyd
  category $(\PP,1,\dot\times,\PP_{\dot T},J,\actionTensor)$ with
  homogeneous countable coproducts, together with the function
  $q_{\psi,\phi,m} \colon \PP_{\dot T}(\psi,\phi)(m) \to
  \VV_T(p\psi,p\phi)$ defined by $q_{\psi,\phi,m}(f)=pf$ is a
  \ghlstr{} over $(\VV,1,\times,\VV_T,I,\ast)$ and $p$.
\end{theorem}
\begin{appendixproof}
  Checking the equalities in Definition \ref{def:logstr} is routine.
  We only check the ex-falso-quodlibet.
  For all $\phi$ and $m$, the mapping
  $q_{\bot_X,\phi,m} \colon \PP_{\dot T}(\bot_X,\phi)(m) \to
  \VV_T(X,p\phi)$ is bijective:
  \begin{displaymath}
    \PP_{\dot T}(\bot_X,\phi)(m)
    = \PP(\bot_X, \dot Tm \phi)
    \cong \VV(X, T p \phi)
    = \VV_T (X, p \phi).
  \end{displaymath}
  Notice that the functor $\bot(X)\triangleq \bot_X$ is a left adjoint to $p$.
\end{appendixproof}

Before seeing examples, we introduce a notation and fibrations for the
assertion
logic.  Let $p:\PP\arrow\VV$ be a fibration for the assertion logic. Below
we use the following notation: for $f\in\VV(I,J)$ and $\psi\in\PP_I$
and $\phi\in\PP_J$, by $f:\psi\darrow \phi$ we mean the statement
``there exists a morphism $\dot f\in\PP(\psi,\phi)$ such that
$p\dot f=f$''. Such $\dot f$ is unique due to the faithfulness of
$p:\PP\arrow\VV$.
\begin{toappendix}
  Using Example \ref{ex:pred},
  we derive fibrations for assertion logic by change-of-base:
  \[
    \quad
    \xymatrix{
      \ERel \ar[r] \ar[d]_-{\erelfib{}} \pbcorner[ul] & \BRel \ar[d]^{\brelfib{}}\\
      \Set \ar[r]^-{\mathrm{diagonal}} & \Set \times \Set
    }
    \quad
    \xymatrix{
      \BRel \ar[r] \ar[d]_{\brelfib{}} \pbcorner[ul] & \Pred \ar[d]^{\predfib{}} \\
      \Set\times\Set \ar[r]^-{\times} & \Set
    }
    \quad
    \xymatrix{
      \Pred(\Meas) \ar[r] \ar[d]_{\predfib{\Meas}} \pbcorner[ul] & \Pred \ar[d]^{\predfib{}} \\
      \Meas \ar[r]^-{|-|} & \Set
    }
  \]
\end{toappendix}

\begin{example}[{Example \ref{exm:syntax-union-bound-logic}}:
  Union Bound Logic]
  To derive the \ghlstr{} suitable for the semantics of the
  Union Bound Logic discussed in Example
  \ref{exm:syntax-union-bound-logic}, we invoke Theorem
  \ref{theorem:logical_structure:graded_lifting} by letting $p$ be
  $\predfib{\Set}:\Pred\to\Set$ (Example \ref{ex:pred}), $T$ be the
  subdistribution monad $\mathcal D$ and $\dot T$ be the strong
  $(\mathbb{R}_{\geq 0},\leq,0,+)$-graded lifting $\mathcal{U}$ of
  $\mathcal D$ defined by
  $ \mathcal{U}(\delta)(X,P) \teq (\mathcal{D}(X), \{ d ~|~ d(X
  \setminus P) \leq \delta\}) $.
  The induced \ghlstr{} is suitable for the
  semantics of GHL for Union Bound Logic in Example
  \ref{exm:syntax-union-bound-logic}. The soundness of inference rules
  follow from the \ghlstr{} as we have showed in Section
  \ref{subsection:interpretation:GHL}. To complete the semantics of
  GHL for the Union Bound Logic, we give the semantics
  $\langle p \rangle$ of procedures $p \in \availproc s$.
  Example~\ref{exm:syntax-union-bound-logic} already gave a semantic
  condition for these operators:
  \begin{align*}
    \lefteqn{\availproc s(\phi, \beta, \psi)}\\
    &=
      \{ \osample\mu\exprc \mid
      \forall s . s \in \sem{\phi} \implies
      \text{Pr}_{s' \leftarrow \llbracket{\osample\mu\exprc}\rrbracket(s)}[s' \in \sem{\neg \psi}] \leq \beta) \}\\
    &= \{ \osample\mu\exprc \mid
      \sem{\osample\mu\exprc}
      \in \Pred_{\mathcal U}(\sem \phi,\sem \psi)(\beta) \}
  \end{align*}
  For any $\osample\mu\exprc \in C_p(\phi, \beta, \psi)$, the
  interpretation $\langle \osample\mu\exprc \rangle $ is
  $\sem{\osample\mu\exprc}$.
\end{example}
\begin{toappendix}
  \begin{example}[A continuous model of the Union Bound Logic.]
    We can also define a continuous model of the Union Bound Logic.
    Let $\mathcal{G}$ be the subprobabilistic variant of the Giry monad (called the sub-Giry monad) \cite{10.1007/BFb0092872} on the category $\Meas$ of
    measurable spaces and measurable functions.
    In a similar way as $\mathcal{U}$, we also define a strong $([0,\infty],+,0,\leq)$-graded lifting $\mathcal{V}$
    of $\mathcal{G}$ along the fibration $\predfib\Meas \colon \Pred(\Meas) \to \Meas$ for assertion logic.
    \begin{align*}
      \mathcal{V}^\delta(X,P)
      & =  \label{eq:gradedmonad;def3} (\mathcal{G}(X), \{ \mu \in \mathcal{G}(X) | \forall A \subseteq_{\text{measurable}}X.~ P \subseteq A \implies \mu(X \setminus A) \leq \delta\}) \tag{L1'} \\
      &=  \label{eq:gradedmonad;def4} \bigcap \{\inverse{(f^\sharp)}(S^{\delta' + \delta}) |f \colon (X,P) \dto S^{\delta'}~ \text{measurable},\delta' \in \mathbb{R}_{\geq 0}\} \tag{L2'}
    \end{align*}
    [$\mathrm{(\ref{eq:gradedmonad;def3})} \subseteq
    \mathrm{(\ref{eq:gradedmonad;def4})}$]
    Suppose $\mu \in \mathrm{(\ref{eq:gradedmonad;def3})}$.
    For any measurable $f \colon (X,P) \dto S^{\delta'}$, we should have $P \subseteq X \setminus f^{-1}([\delta', 1])$.
    Therefore $\mu(f^{-1}([\delta', 1])) \leq \delta$.
    Hence $f^\sharp (\mu) \leq \mu(f^{-1}([\delta', 1])) +  \delta' \leq \delta + \delta'$.
    Hence, $\mu \in \mathrm{(\ref{eq:gradedmonad;def4})}$.
    [$\mathrm{(\ref{eq:gradedmonad;def3})} \supseteq
    \mathrm{(\ref{eq:gradedmonad;def4})}$]
    For any measurable subset $A$ with $P \subseteq A$, define a measurable function
    $f_A \colon (X,P) \dto S^{0}$ by $f(x) = 0$ when $x \in A$ and $f(x) = 1$ otherwise.
    We then have $\delta \geq f_A^\sharp (\mu) = \mu(X \setminus A)$.

    Technically, we handle $X$ and $[0,1]$ as a measurable space and a standard Borel space respectively, and
    we apply the graded version of codensity lifting \cite{DBLP:journals/entcs/Sato16,DBLP:journals/lmcs/KatsumataSU18} instead of the graded $\top\top$-lifting to construct $\mathcal{V}$.
    At all, we obtain a \ghlstr{} induced from $\mathcal{V}$, $\mathcal{G}$ and $\predfib\Meas \colon \Pred(\Meas) \to \Meas$.
    Instantiating GHL with this structure yields a
    continuous version of Union Bound Logic.
  \end{example}
\end{toappendix}

\begin{toappendix}
  \begin{example}[{Example~\ref{exm-syn:simple-cost}}: Simple Cost
    analysis]
    We introduced a cost counting with a natural-number monoid and the
    $\tick{}$ command graded by $1 \in \NN$. To obtain a \ghlstr{} suitable for the GHL for cost counting,
    we invoke Theorem \ref{theorem:logical_structure:graded_lifting} with:
    \begin{itemize}
    \item The fibration for assertion logic given by
      $\predfib{} \colon \Pred \to \Set$.
    \item The monad is the natural-number writer monad $(W, \eta, \mu)$
      on $\Set$ where $W X = X \times \NN$.
    \item The strong $\NN$-graded lifting $\dot W$ is
      given by
      $
      \dot W n (X,P) = (W X, \{ (x, m) \in W X \mid x \in P, m \leq n\})
      $
    \end{itemize}
    The derived \ghlstr{} is suitable for the GHL for simple
    cost analysis.  To complete the semantics of the GHL, we give the
    semantics of the $\mathsf{tick}$ command.  Semantics
    $\sem{\mathsf{tick}} \colon \psi \dto_1 \dot 1$ in $\Pred_{\dot W}$
    of the $\mathsf{tick}$ command is defined by
    $\sem{\mathsf{tick}}(\xi) = (\ast,1)$ for every $\xi \in \mem$, and
    we have the soundness of the axiom
    $\vdash_1 \{ \psi\} { \mathsf{do}~\mathsf{tick}} \{\psi\}$ for every
    $\psi$.
  \end{example}
\end{toappendix}

  \begin{example}[{Example~\ref{exm-syn:program-counter}}: Program
    Counter Security]
    To derive the \ghlstr{} suitable for GHL with program
    counter security, we invoke Theorem
    \ref{theorem:logical_structure:graded_lifting} with:
    \begin{itemize}[leftmargin=0.5em]
    \item The category $\ERel$ of endorelations defined as follows:
    an object $(X,R)$ is a pair of $X \in \Set$ and $R \subseteq X
    \times X$ (i.e. an endorelation $R$ on $X$) and
    an arrow $f \colon (X,R) \to (Y,S)$ is a function $f \colon X \to Y$ such that $(f \times f)(R)\subseteq S$.
    \item The fibration for the assertion logic $\erelfib{}:\ERel\to\Set$ given by $(X,R) \mapsto X$ and $f \mapsto f$.
    \item The writer monad
      $W_s X = X \times \{\otrue,\ofalse\}^*$ on $\Set$ with the monoid of bit
      strings.
    \item The strong $2^*$-graded lifting of $W_s$ along $\erelfib{}:\ERel\to\Set$, given by\\
      $\dot W_s \sigma (X,R) = (W_s X, \{ ((x,\sigma'),(y,\sigma')) ~|~
      (x,y) \in R \land \sigma' \leq \sigma \})$.
    \end{itemize}
    The derived \ghlstr{} is suitable for the semantics of GHL in
    Example~\ref{exm-syn:program-counter}.  To complete the structure of
    the logic, we need to interpret two commands
    $\cfTrue, \cfFalse \in \CExp$ and set the axioms of commands
    $\availcom$.

    First $\sem{\cfTrue}, \sem{\cfFalse} \colon \sem \mem \to 1$ in
    $ \ERel_{\dot W}$ are defined by
    $\sem{\cfTrue} \equiv (\ast,\otrue)$ and
    $\sem{\cfFalse} \equiv (\ast,\ofalse)$.
    Finally, we define $\availcom$ by (recall $\leq$ is prefix
    ordering of strings):
    \[
      \availcom(\psi,\sigma) = \{ \cfTrue \mid \otrue \leq \sigma \} \cup \{
      \cfFalse \mid \ofalse \leq \sigma \}.
    \]
    Note, the graded lifting $\dot W_s \sigma $ relates only the
    pair of $(x,\sigma')$ and $(y,\sigma')$ with common strings of
    control flow.  Hence, the derivation of proof tree of this logic
    forces the target program to have the same control flow under the
    precondition.
  \end{example}

\begin{example}[\Ghlstr{} from the product comonad]
  \label{ex:comonad:logical}
  In the category $\Set$, the functor $CX\teq X \times \NN$ forms a
  coproduct-preserving comonad called the {\em product comonad}.
  The right adjoint $I \colon \Set \to \Set_{C}$ of the coKleisli resolution of $C$
   yields a Freyd category with countable
  coproducts.
  We next introduce a $(\NN,\leq,0\max)$-graded lifting $\dot C$ of
  the comonad $C$ along the fibration
  $\predfib{\Set} \colon \Pred \to \Set$. It is defined by
  $\dot C n (X,P) \teq (CX, \{ (x,m) \in X \times \NN ~|~ x \in P, m
  \geq n \})$. Similarly, we give an $(\NN,\leq,0\max)$-graded Freyd
  category $(J,\actionTensor)$ induced by the graded lifting $\dot C$.
  In this way we obtain a \ghlstr{}.

  By instantiating GHL with the above \ghlstr{}, we obtain a
  program logic useful for reasoning about security levels.  For example,
  when program $P_1$ requires security level $3$ and $P_2$ requires
  security level $7$, the sequential composition $P_1 ; P_2$ requires
  the higher security level $7$ ($= \max (3,7)$).

  We give a simple
  structure for verifying security levels determined by memory access.
  Fix a function $\mathsf{VarLV} \colon \dom(\mctx) \to \NN$ assigning
  security levels to variables.  For any expression $e$, we define its
  required security level
  $\SecLV(e) = \sup\{\VarLV(x) ~|~ x \in
  \mathrm{FV}(e)\}$.  Using this, for each expression $e$ of sort
  $s\in S$ we introduce a procedure $\Psecure e \in \PExp_s$ called
  {\em secured expression}. It returns the value of $e$ if the level
  is high enough, otherwise it returns a meaningless contant:
  \[
    \sem{\Psecure e}(n,\xi) =
    \text{if $n \geq \SecLV(e)$ then $\sem{e}(\xi)$ else a fixed constant $c_s$}.
  \]
  Secured expressions can be introduced through the following
  $\availprocn$:
  \begin{displaymath}
    \availproc s(\phi, l, \psi) =
    \{ \Psecure e \mid e:s, \sem{\Psecure e}\in\Pred_C(\sem \phi, \sem \psi)(l), \mathsf{SecLV}(e) \leq l \}.
  \end{displaymath}
  The \pmonoid{} $(\NN,\leq,0,\max)$ in the above can also be replaced with a
  join semilattice with a least element $(Q,\leq,\bot,\vee)$.
  Thus, GHL can be instantiated to a graded comonadic model of
  security and its associated reasoning.
%


  \begin{toappendix}
    (\ghlstr{} on Example \ref{ex:comonad:logical})
    \begin{lemma}
      (1) the tuple $(\Set,1,\times,\Set_C,I,\ast)$ induced by $C$ forms a Freyd category with coproduct;
      (2) the tuple $(\Pred,\dot 1,\dtimes,\Pred_{\dot C},J,\actionTensor)$ induced by the graded lifting $\dot C$ forms an $(\NN,\max,0,\leq)$-graded Freyd category
      with homogeneous coproduct;
      (3) they form a \ghlstr{}.
    \end{lemma}
    \begin{proof}
      (1) We need to show the equalities on $J$ and $\ast$.
      From the counit law of comonad, we have $I(\id_X) = \id_X$ and $I(g\circ f) = Ig \circ If$.
      We also have:
      \begin{align*}
        & \id_V*\id_X = (\id_V \times \pi_1) \circ \langle \pi_1 \circ \pi_1, \langle \pi_2 \circ \pi_1, \pi_2 \rangle \rangle
          =  \langle \pi_1 \circ \pi_1, \pi_2 \circ \pi_1 \rangle
          = \pi_1 = \id_{V*X}\\
        &(g\circ f)*(i\circ j) = (g\circ f) \times (i\circ \langle j, \pi_2 \rangle) \circ \langle \pi_1 \circ \pi_1, \langle \pi_2 \circ \pi_1, \pi_2 \rangle \rangle\\
        &\qquad = (g \times i) \circ  \langle f \circ \pi_1 \circ \pi_1, \langle j, \pi_2 \rangle \circ \langle \pi_2 \circ \pi_1, \pi_2 \rangle \rangle\\
        &\qquad = (g \times i) \circ  \langle f \circ \pi_1 \circ \pi_1,  \langle j \circ \langle \pi_2 \circ \pi_1, \pi_2  \rangle, \pi_2 \rangle \rangle\\
        &\qquad = (g \times i) \circ  \langle \pi_1 \circ (f \times j) \circ \langle\ \pi_1 \circ \pi_1, \langle \pi_2 \circ \pi_1, \pi_2  \rangle\rangle, \langle \pi_2 \circ (f \times j) \circ \langle\ \pi_1 \circ \pi_1, \langle \pi_2 \circ \pi_1, \pi_2  \rangle\rangle, \pi_2 \rangle \rangle\\
        &\qquad = (g \times i) \circ \langle\ \pi_1 \circ \pi_1, \langle \pi_2 \circ \pi_1, \pi_2  \rangle \rangle \circ \langle\ (f \times j) \circ \langle\ \pi_1 \circ \pi_1, \langle \pi_2 \circ \pi_1, \pi_2  \rangle\rangle, \pi_2  \rangle\\
        &\qquad =(g*i)\circ (f*j),\\
        & f\circ I(l_X) = f \circ \langle \pi_2 \circ \pi_1, \pi_2 \rangle =  \pi_2 \circ \pi_1  \circ \langle\ (\id_1 \times f) \circ \langle\ \pi_1 \circ \pi_1, \langle \pi_2 \circ \pi_1, \pi_2  \rangle\rangle, \pi_2  \rangle = I(l_X)\circ(\id_1*f),\\
        &I(a_{X',Y',Z'})\circ ((f\times g)*h)
          = a_{X',Y',Z'} \circ (f \times g) \times h \circ \langle\ \pi_1 \circ \pi_1, \langle \pi_2 \circ \pi_1, \pi_2  \rangle\rangle\\
    	&\qquad  = a_{X',Y',Z'} \circ (f \times g) \times h \circ a_{X\times Y,Z,\NN}
       = f \times (g \times h) \circ (\mathrm{id}_X \times a_{Y,Z,\NN}) \circ a_{X,Y \times Z,\NN} \circ (a_{X,Y,Z} \times \mathrm{id}_\NN)\\
    	&\qquad = (f\ast (g\ast h)) \circ \langle a_{X,Y,Z}\circ \pi_1,\pi_2 \rangle = (f\ast (g\ast h)) \circ I( a_{X,Y,Z})\\
        & I(f\times g) = (f\times g) \circ \pi_1 = (f \times g \circ \pi_1) \circ  \langle \pi_1\circ\pi_1, \langle  \pi_2 \circ \pi_1, \pi_2 \rangle\rangle=  f \ast Ig
      \end{align*}

      Next, we show the coproducts in $\Set_C$.
      For coprojections $\iota_i \colon X_i \to \coprod_{i \in I}X_i$, we define $I(\iota_i) \colon C X_i \to \coprod_{i \in I}X_i$.
      For given $f_i \colon CX_i \to X$ ($i \in I$), we define $[f_i]^\ast_{i \in I} = [f_i]_{i \in I} \circ \mathrm{dist} \colon C(\coprod_{i \in I} X_i) \to X$ where $\mathrm{dist} \colon (\coprod_{i \in I} X_i) \times \NN \to \coprod_{i \in I} (X_i \times \NN)$ is the inverse of $[\langle \iota_i \circ \pi_1, \pi_2 \rangle]_{i \in I}$.
      We show that $I(\iota_i) $ and $[f_i]^\ast_{i \in I}$ satisfy the conditions of coprojections and cotuplings:
      \begin{align*}
        &[f_i]^\ast_{i \in I} \circ I(\iota_i) = [f_i]_{i \in I} \circ \mathrm{dist}  \circ (\iota_i \times \mathrm{id}_\NN) = [f_i]_{i \in I} \circ \iota_i =  f_i\\
        &[I(\iota_i)]^\ast_{i \in I} = [\iota_i \times \mathrm{id}_\NN]_{i \in I}  \circ \mathrm{dist}  = \mathrm{id}_{(\coprod_{i \in I} X_i) \times \NN } \\
        &g \circ [f_i]^\ast_{i \in I}
          = g \circ \langle  [f_i]_{i \in I} \circ \mathrm{dist} ,\pi_2 \rangle = [g \circ \langle f_i, \pi_2 \rangle ]_{i \in I} \circ \mathrm{dist}
          =  [g \circ f_i]^\ast_{i \in I}
      \end{align*}
      Hence, $(\Set,1,\times,\Set_C,I,\ast)$ is a Freyd category with products.

      (2) Since $\dot C$ is a lifting of $C$, we inherit the structure of the Freyd category with products (The mapping $J$ and $\actionTensor$ are the same as $I$ and $\ast$, as mappings of arrows).
      We need to check that the gradings of $J$ and $\actionTensor$ are well-defined, and the mappings $\sqsubseteq^{m \leq n} \colon \dot Cn \to \dot Cm$ behave well.

      First, we check that we have $I(f) \colon \dot C0(X,P) \dto (Y,P)$ for given $f \colon (X,P) \to (Y,Q)$.
      Let $(x,m)$ in the predicate part of $\dot C0(X,P)$.
      We then have $x \in P$ and $m \geq 0$.
      Since $f \colon (X,P) \to (Y,Q)$, $f(x) \in Q$.
      Hence $I(f)(x,m) = f \circ \pi_1(x,m) = f(x) \in Q$.
      This implies $I(f) \colon \dot C0(X,P) \dto (Y,P)$.

      Next, for all $f \colon (X,P) \dto (X',P')$ and $g \colon \dot Cm(Y,Q) \to (Y',Q')$, we have
      $f \ast g \colon \dot Cm ( (X,P) \dtimes (Y,Q) ) \dto (X',P') \dtimes (Y',Q')$.
      Let $((x,y),n)$ in the predicate part of $\dot Cm ( (X,P) \dtimes (Y,Q) )$, that is,
      $x \in P$, $y \in Q$ and $n \geq m$.
      Since  $f \colon (X,P) \dto (X',P')$ and $g \colon \dot Cm(Y,Q) \to (Y',Q')$, we have $f(x) \in P'$ and $g(y,n) \in Q'$.
      Hence $(f \ast g)((x,y),n) = (f(x),g(y,n)) \in P' \times Q'$.

      We can inherit the conditions of coprojections and cotuplings on $C$
      to ones for homogeneous coproducts, but we need to check the gradings of coprojections and cotuplings.
      It is easy to see that $I(\iota_i) \colon C X_i \to \coprod_{i \in I}X_i$
      forms $I(\iota_i) \colon \dot C 0 (X_i,P_i) \to \dot \coprod_{i \in I} (X_i,P_i)$.
      Next, we show that when $f_i \colon \dot Cm (X_i,P_i) \dto (X,P)$ ($i \in I$), we have
      $[f_i]^\ast_{i \in I} \colon \dot C m (\dot \coprod_{i \in I} (X_i,P_i)) \dto (X,P)$.
      Let $(x,n)$ in the predicate part of $\dot C m (\dot \coprod_{i \in I} (X_i,P_i))$.
      Then, $n \geq m$ and there is a unique $i \in I$ such that $x \in X$ and $x \in P_i$.
      Hence $[f_i]^\ast_{i \in I} (x,n) = f_i(x,n) \in P$.

      Finally, we check the law of homogeneous coproduct on $\sqsubseteq^{m \leq n} \colon \dot Cn \to \dot Cm$.
      We first have  $\mathrm{dist} \colon \dot Cm( \dot \coprod_{i \in I} (X_i,P_i)) \dto \dot \coprod_{i \in I} \dot Cm (X_i,P_i)$.
      To show this, we need to check
      \[
        [\langle \iota_i \circ \pi_1, \pi_2 \rangle]_{i \in I}^{-1}(\dot Cm( \dot \coprod_{i \in I} (X_i,P_i)) \subseteq (\dot \coprod_{i \in I} \dot Cm (X_i,P_i) ).
      \]
      Let $(x,n) \in \coprod_{i \in I} (X_i \times \NN)$ with $x \in X_i$.
      We have $[\langle \iota_i \circ \pi_1, \pi_2 \rangle]_{i \in I}(x,n) = (\iota_i(x),n)$.
      If $(\iota_i(x),n) \in \dot Cm( \dot \coprod_{i \in I} (X_i,P_i)$, $x \in P_i$ and $n \geq m$.
      Then $(x,n) \in \dot \coprod_{i \in I} Cm (X_i,P_i)$.
      Next, we show $\mathrm{dist} \circ  \sqsubseteq^{m \leq n} = \dot \coprod_{i \in I} \sqsubseteq^{m \leq n}  \circ \mathrm{dist} $.
      To show this, it suffices to check
      $\sqsubseteq^{m \leq n}  \circ [\langle \iota_i \circ \pi_1, \pi_2 \rangle] =[\langle \iota_i \circ \pi_1, \pi_2 \rangle] \circ  \dot \coprod_{i \in I} \sqsubseteq^{m \leq n}$, but
      it is straightforward.
      Then we can check the remaining law of homogeneous coproduct:
      \[
        [\uparrow_{m}^{n} f_i]^\ast_{i \in I}
        =
        [f_i \circ \sqsubseteq^{m \leq n} ]_{i \in I} \circ \mathrm{dist}
        =
        [f_i]_{i \in I} \circ \sqsubseteq^{m \leq n} \circ \mathrm{dist}
        =
        [f_i]_{i \in I} \circ \mathrm{dist} \circ \dot \coprod_{i \in I} \sqsubseteq^{m \leq n}
        =
        \uparrow_{m}^{n} [f_i]^\ast_{i \in I}
      \]
      Therefore, $(\Pred,\dot 1,\dtimes,\Pred_{\dot C},J,\actionTensor)$ forms a Freyd category with homogeneous coproducts.

      (3) The functor $q \colon \Pred_{\dot C} \to \Set_C$ is defined by $f \mapsto pf$.
      The ex falso quodlibet is obvious.
      The rest of proof is straightforward.
    \end{proof}
  \end{toappendix}

\end{example}

\begin{toappendix}
  \begin{example}[A Program Logic for Differential Privacy]
    Program logics for reasoning about differential privacy can be seen as instantiations of GHL with a relational model.
    For example, the continuous version of the logic apRHL by Sato~\cite{DBLP:journals/entcs/Sato16} is based on the strong $[0,\infty] \times [0,\infty]$-graded lifting $\mathcal{G}^{\top\top}$ of $\mathcal{G} \times \mathcal{G}$, which is a graded extension of the codensity lifting given in \cite[Section 4.3.2]{DBLP:journals/lmcs/KatsumataSU18}:
    \begin{align*}
      \mathcal{G}^{\top\top(\varepsilon,\delta)}(X,Y,R)
      & = (\mathcal{G}(X),\mathcal{G}(Y), \{ (\mu_1,\mu_2) | \forall A \in \Sigma_X, B \in \Sigma_Y.~R(A) \subseteq B {\implies} \mu_1(A) \leq e^\varepsilon \mu_2(B) + \delta \})\\
      &=  \bigcap \{\inverse{(f^\sharp \times g^\sharp)}(S^{\varepsilon + \varepsilon', \delta + \delta'}) | f \colon (X,Y,R) \dto S^{\varepsilon',\delta',}~ \text{measurable},\varepsilon',\delta' \in [0,\infty] \}\\
      & \qquad \text{ where } S^{\varepsilon,\delta} = ([0,1],[0,1], \{ (r,s) | r \leq e^\varepsilon s + \delta \})
    \end{align*}
    along the fibration $\brelfib{\Meas} \colon \BRel(\Meas) \to \Meas \times \Meas$ for assertion logic defined by the following change-of-base
    \[
      \quad
      \xymatrix{
        \BRel(\Meas) \ar[r] \ar[d]_{\brelfib\Meas} \pbcorner[ul] & \BRel \ar[d]^{\brelfib{}} \\
        \Meas \times \Meas \ar[r]^-{(|-|,|-|)} & \Set\times\Set
      }
    \]
    Thus, we have a \ghlstr{} corresponding the semantic model of continuous apRHL:
    \[
      \xymatrix{
	\BRel(\Meas) \ar[rr]^{(J,\actionTensor)} \ar[d]_{\brelfib{\Meas} } & & \BRel(\Meas)_{\mathcal{G}^{\top\top}} \ar[d]^{q} \\
	\Meas \times \Meas \ar[rr]^{(I,\ast)} & & {\Meas \times \Meas}_{\mathcal{G} \times \mathcal{G}}
      }
    \]

    We can introduce distributions as procedures, e.g. $\PExp ::= \mathsf{Gauss}(e_\mu,e_\sigma) \mid \mathsf{Laplace}(e_\mu,e_\lambda) \mid \cdots$, and
    give the following partial definition of $C_p$ corresponding the [rand] rule in~\cite{DBLP:journals/entcs/Sato16}:
    \[
      C_p(\phi,(\varepsilon,\delta), x \langle 1 \rangle =  y \langle 2 \rangle)
      = \{ p | \sem{p} \colon \sem \phi \dto_{(\varepsilon,\delta)} \sem{x \langle 1 \rangle =  y \langle 2 \rangle} \}
    \]
  \end{example}
\end{toappendix}

\begin{toappendix}
  \begin{example}[Span-apRHL for Relaxations of Differential Privacy]
    \label{ex:span-semantics}
    GHL allows the flexibility of separating the model of logical
    assertions from the one of judgments.  We demonstrate this
    flexibility by instantiating it to capture
    span-apRHL~\cite{DBLP:conf/lics/SatoBGHK19}, a variant of the
    continuous apRHL useful to formally reason about relaxations of
    differential privacy.  In this logic, assertions are modeled as
    binary relation, while judgments are modeled as morphism between
    spans, which are richer than just relations.  The semantic models
    of span-apRHL given in \cite{DBLP:conf/lics/SatoBGHK19} can be
    regarded the following \ghlstr{}.
    \[
      \xymatrix{ \BRel(\Meas)
        \ar[r]_-{K} \ar[d]^-{\brelfib{\Meas}} &
        \mathbf{Span}(\Meas)|_{\BRel(\Meas)} \ar [r]^-{(\hat J,\hat
          \actionTensor)} \ar[d]^-{\hat
          p|_{\mathbf{Span}(\Meas)|_{\BRel(\Meas)}}}
        &{\mathbf{Span}(\Meas) }_{(- )^{\sharp
            (-,-)\Delta}}|_{\BRel(\Meas)} \ar[d]|{ q = \hat q
          |_{{\mathbf{Span}(\Meas) }_{(- )^{\sharp
                (-,-)\Delta}}|_{\BRel(\Meas)}}}
        \\
        \Meas\times\Meas \ar[r] & \Meas\times\Meas \ar[r]_-{(I, \ast)}
        & \Meas\times\Meas_{\mathcal{G}\times\mathcal{G}} }
    \]
    Here, $\brelfib\Meas:\BRel(\Meas) \to \Meas \times \Meas$ is a
    fibration for assertion logic; $\mathbf{Span}(\Meas)$ is the
    category of measurable spans $(X \leftarrow R \rightarrow Y)$;
    $\hat p \colon \mathbf{Span}(\Meas) \to \Meas \times \Meas$ is a
    functor sending $(X \leftarrow R \rightarrow Y)$ to the underlying
    spaces $(X,Y)$; $(- )^{\sharp (-,-)\Delta}$ is an approximate
    span-lifting for an $M$-graded family
    $\Delta = \{\Delta^m\}_{m \in M}$ of reflexive composable
    divergences.  The approximate span-lifting is a strong
    $M \times \overline{\RR}$-graded lifting of the product monad
    $\mathcal{G}\times\mathcal{G}$ along $\hat p$, where $\mathcal G$
    is the subprobabilistic variant of the {\em Giry monad}
    \cite{10.1007/BFb0092872}.

    The functor $K \colon \BRel(\Meas) \to \mathbf{Span}(\Meas)$
    embeds binary relations $(X,Y,R)$ to measurable spans
    $(X \leftarrow R \rightarrow Y)$ where $R$ is regarded as the
    subspace of $X \times Y$.  This functor preserves products and
    coproducts strictly.  The codomain of $K$ can be restricted to the
    image $\mathbf{Span}(\Meas)|_{\BRel(\Meas)}$, the subcategory
    whose objects are restricted to binary relations embedded by $K$.

    The $M \times \overline{\RR}$-graded category
    ${\mathbf{Span}(\Meas) }_{(- )^{\sharp
        (-,-)\Delta}}|_{\BRel(\Meas)}$ is the restriction of the
    Kleisli graded category
    ${\mathbf{Span}(\Meas) }_{(- )^{\sharp (-,-)\Delta}}$ of
    approximate span-lifting whose objects are restricted to binary
    relations embedded by $K$.

    The Freyd category $(I,\ast)$ with product is induced from the
    structure of strong monad $\mathcal{G}\times\mathcal{G}$.  The
    $M \times \overline{R}$-graded Freyd category $(J,\actionTensor)$
    with homogeneous coproducts is given by composing the embedding
    $K$ and the $M \times \overline{R}$-graded
    $(\hat J, \hat \actionTensor)$ with homogeneous coproduct induced
    by the structure of $M \times \overline{\RR}$-graded monad
    $(- )^{\sharp (-,-)\Delta}$ on $\mathbf{Span}(\Meas)$.  Namely, we
    define $J = \hat J \circ K$ and
    $ f \actionTensor g = (K f) \hat \actionTensor g$. It is
    straightforward to check the conditions of the \ghlstr{}.

    By instantiating GHL with this \ghlstr{}, we obtain a part
    of span-apRHL for a divergence $\Delta$.
    We can introduce distributions as procedures, e.g.
    $\PExp ::= \mathsf{Gauss}(e_\mu,e_\sigma) \mid
    \mathsf{Laplace}(e_\mu,e_\lambda) \mid \cdots$, and give the
    following partial definition of $C_p$ as follows ($f$ is called
    witness function):
    \[
      C_p(\phi,(m,\delta),\psi) = \{ p | \exists f \text{ s.t. }
      (\sem{p}, f) \in ({\mathbf{Span}(\Meas) }_{(- )^{\sharp
          (-,-)\Delta}}|_{\BRel(\Meas)}) (\sem
      \phi,\sem{\psi})(m,\delta) \}.
    \]
  \end{example}
\end{toappendix}



\section{Related Work}


Several works have studied abstract semantics of Hoare Logic.
Martin et al.~\cite{10.1007/11874683_33} give a categorical framework
based on traced symmetric monoidal closed categories. They also show
that their framework can handle extensions such as separation
logic. However their framework does not directly model effects and it
cannot accommodate grading as is.
Goncharov and Shr\"{o}der~\cite{DBLP:conf/lics/GoncharovS13} study a Hoare Logic to reason in a
generic way about programs with side effects. Their logic and
underlying semantics is based on an order-enriched monad and they show
a relative completeness result. Similarly, Hasuo~\cite{HASUO20152}
studies an abstract weakest precondition semantics based on
order-enriched monad. A similar categorical model has also been used
by Jacobs~\cite{Jacobs15} to study the Dijkstra monad and the Hoare monad.
In the logic by Goncharov and Shr\"{o}der~\cite{DBLP:conf/lics/GoncharovS13}  effects are encapsulated in monadic types, while the weakest precondition semantics by Hasuo~\cite{HASUO20152} and the semantics by Jacobs~\cite{Jacobs15} have no underlying calculus. Moreover, none of them is graded.
Maillard et al.~\cite{DBLP:journals/pacmpl/MaillardAAMHRT19} study a semantics
framework based on the Dijkstra monad for program verification. Their
framework enables reasoning about different side effects and it separates specification from computation. Their Dijkstra monad has a flavor of grading but the  structure they use is more complex than a \pmonoid{}.
Maillard et al.~\cite{DBLP:journals/pacmpl/MaillardHRM20} focus on relational program logics for effectful computations. They show how these logics can be derived in a relational dependent  type theory, but their logics are not graded.

As we discussed in the introduction, several works have used
\emph{grading} structures similar to the one we study in this paper,
although often with different names.
Katsumata studied monads graded by a
\pmonoid{} as a semantic model for effects system~\cite{DBLP:conf/popl/Katsumata14}. A similar
approach has also been studied elsewhere~\cite{DBLP:conf/birthday/MycroftOP16,DBLP:journals/corr/OrchardPM14}.
Formal categorical
properties of graded monads are pursued by Fujii et
al.~\cite{fujii2016towards}.
Zhang defines a notion of \emph{graded category}, but it differs
to ours, and is instead closer to a definition of a graded monad~\cite{zhang1996twisted}.
As we showed in Section~\ref{sec:graded-category}, graded
categories can be constructed both by monads and
comonads graded by a \pmonoid{}, and it
can also capture graded structures that do not arise from either of
them.  Milius et al.~\cite{milius2015generic} also studied monads graded by a
\pmonoid{} in the context of trace semantics where the grading
represents a notion of depth corresponding to trace length. Exploring
whether there is a generalization of our work to traces is an
interesting future work.  

Various works study comonads graded with a semiring structure as a
semantic model of contextual computations captured by means of type
systems~\cite{DBLP:conf/esop/BrunelGMZ14,DBLP:conf/esop/GhicaS14,DBLP:conf/icfp/PetricekOM14}. In
contrast, our graded comonads are graded by a pomonoid. The
additive  structure of the semiring in those works is needed to merge
the gradings of different instances of the same variable. This is
natural for the $\lambda$-calculus where the context represent
multiple inputs, but there is only one conclusion (output). Here
instead, we focus on an imperative language. So, we have only one
input, the starting memory, and one output, the updated
memory. Therefore, it is natural to have just the multiplicative
structure of the semiring as a \pmonoid{}. The categorical
axiomatics of semiring-graded comonads are studied by Katsumata from the double-category theoretic perspective~\cite{DBLP:conf/fossacs/Katsumata18}.

Apart from graded monads, several generalizations of monads has been
proposed. Atkey introduces {\em parameterized monads} and
corresponding {\em parameterized Freyd categories}~\cite{Atkey09}, demonstrating
that parameterized monads naturally model effectful computations
with preconditions and postconditions.  Tate
defines {\em productors} with
composability of effectful computations controlled by
a relational `effector' structure~\cite{Tate13}.
Orchard et al. define
\emph{category-graded monads}, generalizing graded and parameterised
monads via lax functors and sketch a model of Union Bound Logic in
this setting (but predicates and graded-predicate interaction are not
modelled, as they are here)~\cite{DBLP:journals/corr/abs-2001-10274}.
Interesting future
work is to combine these
general models of computational effects with Hoare logic.

\section{Conclusion}
We have presented a Graded Hoare Logic as a parameterisable framework
for reasoning about programs and their side effects, and studied its
categorical semantics. The key guiding idea is that grading can be seen as a refinement of effectful computations. This has brought us naturally to graded categories but to fully internalize this refinement idea we further introduced the new notion of graded Freyd categories. To show the generality of our framework we have shown how different examples are naturally captured by it.

We conclude with some reflections on possible future work.


\paragraph{Future work}
Carbonneaux et al. present a quantitative verification approach for amortized cost analysis via a Hoare logic augmented with multivariate
quantities associated to program
variables~\cite{DBLP:conf/pldi/Carbonneaux0S15}.  Judgments $\vdash
\{\Gamma; Q\} S \{\Gamma'; Q'\}$ have pre- and post-conditions $\Gamma$
and $\Gamma'$ and potential functions $Q$ and $Q'$.
%
Their approach can be mapped to GHL
with a grading monoid representing how the potential functions change.
However, the multivariate
nature of the analysis requires a more fine-grained connection
between the structure of the memory and the structure of grades,
which have not been developed yet. We leave this for future work.

GHL allows us to capture the dependencies between assertions and
grading that graded program logics usually use. However, some graded
systems (e.g.~\cite{BartheGAHRS15}) use more explicit dependencies by
allowing grade variables---which are also used for grading
polymorphism. We plan to explore this direction in future work.

The setting of graded categories in this work subsumes both graded
monads and graded comonads and allows flexibility in the
model. However, most of our examples in Section~\ref{sec:instances}
are related to graded monads.  The literature contains various graded
comonad models of data-flow properties: like liveness
analysis~\cite{DBLP:conf/icfp/PetricekOM14},
sensitivities~\cite{DBLP:conf/esop/BrunelGMZ14}, timing and
scheduling~\cite{DBLP:conf/esop/GhicaS14}, and information-flow
control~\cite{DBLP:journals/pacmpl/OrchardLE19}.  Future work is to
investigate how these structures could be adopted to GHL for reasoning
about 
programs.



\paragraph{Acknowledgements}
Katsumata and Sato carried out this research supported by ERATO
HASUO Metamathematics for Systems Design Project (No. JPMJER1603), JST.
Orchard is supported by EPSRC grant EP/T013516/1.
Gaboardi is supported by the National Science Foundation under Grant No. 2040222.

\bibliographystyle{splncs04}
\bibliography{references,mfps20}


\vfill

{\small\medskip\noindent{\bf Open Access} This chapter is licensed under the terms of the Creative Commons\break Attribution 4.0 International License (\url{http://creativecommons.org/licenses/by/4.0/}), which permits use, sharing, adaptation, distribution and reproduction in any medium or format, as long as you give appropriate credit to the original author(s) and the source, provide a link to the Creative Commons license and indicate if changes were made.}

{\small \spaceskip .28em plus .1em minus .1em The images or other third party material in this chapter are included in the chapter's Creative Commons license, unless indicated otherwise in a credit line to the material.~If material is not included in the chapter's Creative Commons license and your intended\break use is not permitted by statutory regulation or exceeds the permitted use, you will need to obtain permission directly from the copyright holder.}

\medskip\noindent\includegraphics{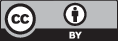}

\ifextended
\newpage
\appendix

\fi

\end{document}